\documentclass[11pt]{article}

\usepackage[T1]{fontenc}
\usepackage{lmodern}
\usepackage[margin=1in]{geometry}
\usepackage{amsmath,amssymb,amsthm,mathtools}
\usepackage{bm}
\usepackage{array}
\usepackage{booktabs}
\usepackage{graphicx}
\usepackage{tikz}
\usepackage{xcolor}
\usepackage{etoolbox}
\usepackage{placeins}
\usepackage{needspace}
\usepackage{hyperref}
\usepackage{bookmark}
\usepackage[nameinlink,capitalize,noabbrev]{cleveref}
\usetikzlibrary{arrows.meta,positioning}

\hypersetup{
  unicode=true,
  pdfencoding=auto,
  psdextra=true,
  colorlinks=true,
  linkcolor=blue!60!black,
  citecolor=blue!60!black,
  urlcolor=blue!60!black
}
\theoremstyle{definition}
\newtheorem{theorem}{Theorem}
\newtheorem{proposition}{Proposition}
\newtheorem{lemma}{Lemma}
\newtheorem{remark}{Remark}
\newtheorem{definition}{Definition}

\newcommand{\theoremqedsymbol}{\ensuremath{\square}}
\AtEndEnvironment{theorem}{\hfill\theoremqedsymbol}
\AtEndEnvironment{proposition}{\hfill\theoremqedsymbol}
\AtEndEnvironment{lemma}{\hfill\theoremqedsymbol}
\AtEndEnvironment{definition}{\hfill\theoremqedsymbol}
\AtEndEnvironment{remark}{\hfill\theoremqedsymbol}
\AtEndEnvironment{example}{\hfill\theoremqedsymbol}
\AtEndEnvironment{discussion}{\hfill\theoremqedsymbol}
\crefname{discussion}{discussion}{discussions}
\Crefname{discussion}{Discussion}{Discussions}

\newcommand{\F}{\mathbb{F}}
\newcommand{\AZ}{A_Z}
\newcommand{\AX}{A_X}
\newcommand{\AD}{A_\Delta}
\newcommand{\CZ}{C_Z}
\newcommand{\CX}{C_X}
\newcommand{\Row}{\operatorname{Row}}
\newcommand{\Ker}{\operatorname{Ker}}
\newcommand{\vf}{\bm{f}}
\newcommand{\vg}{\bm{g}}
\newcommand{\ve}{\bm{e}}
\newcommand{\vt}{\bm{t}}

\title{Spatially Coupled MacKay--Neal/Hsu--Anastasopoulos CSS Codes Achieve the Quantum-Erasure Hashing Bound by Seeded BP Decoding}
\author{Kenta Kasai\\Institute of Science Tokyo\\\texttt{kenta@ict.eng.isct.ac.jp}}
\date{}

\begin{document}
\maketitle

\begin{abstract}
We study hard-erasure belief-propagation (BP), understood as the sum-product algorithm on erasure messages, for a punctured MacKay--Neal/Hsu--Anastasopoulos (MN/HA) representation of Calderbank--Shor--Steane (CSS) codes.  For every integer degree triple $2\leq j_Z<j_X<k\leq60$, we prove that the constituent potential thresholds are $j_Z/k$ and $1-j_X/k$.  The proof combines classical HA--MN duality with a completed exact rational certificate for all 1711 MN degree pairs in this range.  The certificate establishes strict positivity on every physical nontrivial fixed-point branch, at every erasure probability in $[0,1]$, and has an independent integer-arithmetic verifier.  A reduction of the Z-side recursion and a fixed-channel vector-potential argument then prove seeded density-evolution convergence below the smaller constituent threshold for sufficiently large coupling width and an ideal seed interval at least that wide.  Under $j_Z+j_X=k$, this threshold equals the quantum-erasure hashing-bound parameter determined by the CSS design rate.  In particular, the degree triple $(4,8,12)$ has potential threshold $1/3$ without a positivity hypothesis.  The numerical example at erasure probability $0.3325$ reproduces the resulting decoding waves.  The theorem uses an ideal auxiliary-message seed; its finite-code realization and logical block-error convergence are separate questions.

\end{abstract}

\section{Introduction}

Calderbank--Shor--Steane (CSS) codes are a standard way to build quantum stabilizer codes from an orthogonality relation between two classical linear codes \cite{CalderbankShor1996,Steane1996}.  On the classical side, low-density parity-check (LDPC) codes make local graph-based decoding possible \cite{Gallager1962,RichardsonUrbanke2008}.  The MacKay--Neal (MN) and Hsu--Anastasopoulos (HA) ensembles use punctured sparse representations to obtain bounded-degree graphical constructions with capacity-oriented decoding properties \cite{MacKayNeal1996,HsuAnastasopoulos2010}.

Quantum LDPC code construction has developed in parallel with the classical sparse-graph coding literature.  Hypergraph-product codes gave positive-rate quantum LDPC families with distance proportional to the square root of the block length \cite{TillichZemor2014}, and lifted-product constructions later gave asymptotically good quantum LDPC codes \cite{PanteleevKalachev2022}.  Spatial coupling has also been studied directly for quantum LDPC codes, including spatially coupled quasi-cyclic quantum LDPC codes, entanglement-assisted spatially coupled quantum LDPC codes, and algebraic spatially coupled quantum LDPC constructions \cite{HagiwaraKasaiImaiSakaniwa2011,AndriyanovaMauriceTillich2012,YangCalderbank2025}.  These construction results motivate decoding analyses for CSS codes, but they do not by themselves provide a density-evolution (DE) theorem for the punctured MN/HA representation considered here.

Before the terminology of spatial coupling became standard, the same construction principle appeared in the study of low-density parity-check (LDPC) convolutional codes: periodic low-density convolutional parity-check matrices were introduced in \cite{FelstromZigangirov1999}, and terminated LDPC convolutional ensembles were shown by density evolution to have belief-propagation (BP) thresholds close to capacity in \cite{LentmaierSridharanZigangirovCostello2005,LentmaierSridharanCostelloZigangirov2010}.  Protograph-based LDPC convolutional and spatially coupled LDPC ensembles were then developed as a structured design framework with distance and threshold analysis \cite{MitchellPusaneZigangirovCostello2008,MitchellLentmaierCostello2015}.  Spatial coupling was also formulated in the classical LDPC setting as a method that lets BP decoding attain the maximum-a-posteriori (MAP) or area-threshold performance of the uncoupled system \cite{KudekarRichardsonUrbanke2011}.  The vector-potential method extends the same proof strategy to DE systems with multiple message types \cite{YedlaVector2012}.  For MN/HA-type ensembles, prior work proceeds from threshold improvement and asymptotic analysis on the binary erasure channel (BEC), to capacity-achieving bounded-degree cases, to symmetric-information-rate (SIR) achievability over generalized erasure channels (GECs) \cite{KasaiSakaniwa2011,MitchellKasaiLentmaierCostello2012,ObataJianKasaiPfister2013,OkazakiKasai2014,FukushimaOkazakiKasai2015}.  Thus, from the classical viewpoint, the natural expectation is that spatial coupling should expose the MAP performance of MN/HA-type systems to BP decoding.

This expectation has a direct CSS analogue.  If the Z-side and X-side MN/HA constituents have MAP thresholds matching their respective capacity or area limits, then using those constituents in a CSS sparse representation suggests that spatial coupling should make seeded BP decoding reach the quantum-erasure hashing-bound parameter.  The two Pauli components share the same erasure locations, but the hard-erasure decoder considered here does not exchange information between the two constituents.  The five-message description therefore consists of separate classical erasure recursions.  The CSS-specific requirement is their compatibility with the nested sparse representation.  The positivity required for these classical constituents is proved below for every $2\leq j_Z<j_X<k\leq60$ by exact rational verification and HA--MN duality.  The resulting seeded DE theorem has no unproved positivity hypothesis.

The finite-degree CSS construction in \cite{Kasai2026FiniteDegree} uses MN/HA-type punctured sparse representations as nested structures for building CSS pairs.  Its minimum-distance result is an existence statement for the finite random code ensemble: with suitable finite degrees, the resulting CSS codes have positive relative minimum distance, and the achievable rate--distance tradeoff reaches the quantum Gilbert--Varshamov benchmark.  This distance statement controls the code family itself, but it does not specify an iterative erasure decoder or its threshold.  This paper isolates the hard-erasure decoding problem associated with the same sparse representation over the quantum erasure channel.  On erased coordinates, the two Pauli components remain unresolved, so the Z-side and X-side sparse decoding problems are described by a five-message recursion driven by the same erasure probability.

We derive the uncoupled and position-dependent recursions, identify their constituent potentials, and prove convergence for ideal seeded coupling at the stated degrees.  The Z-side auxiliary channel message is eliminated before applying the vector-potential argument.  The finite bound $k\leq60$ is part of the proved theorem and includes every degree example in this paper.  The verification covers all intermediate degree pairs, not only the plotted examples.

\section{Sparse CSS Erasure Models}\label{sec:model}

This section first fixes the CSS notation and the quantum-erasure decoding model, and then defines the ensemble through punctured sparse representations.  We begin by briefly reviewing the finite-degree CSS construction of \cite{Kasai2026FiniteDegree}, keeping only the notation needed for the later decoding analysis.  All vectors are column vectors, and \(\bm0\) denotes the zero column vector of the context-dependent length.  We write \(\Row(M)\) for the row span of \(M\), represented as a subspace of column vectors.

Start from a general CSS code specified by two binary check matrices
\[
      H_X\in\F_2^{m_X\times N},\qquad
      H_Z\in\F_2^{m_Z\times N},\qquad
      H_XH_Z^T=\bm{0}.
\]
Rows of \(H_X\) are interpreted as X-type stabilizers, and rows of \(H_Z\) as Z-type stabilizers.  Write a binary Pauli error as
\((\ve_X,\ve_Z)\in\F_2^N\times\F_2^N\).  The measured syndromes are
\[
      \bm{\sigma}_X=H_X\ve_Z,\qquad
      \bm{\sigma}_Z=H_Z\ve_X .
\]

The decoder is given the erasure state \(S=(S_1,\ldots,S_N)\).  If \(S_i=0\), the coordinate is known and \((e_{X,i},e_{Z,i})=(0,0)\).  If \(S_i=1\), the erased coordinate leaves \(e_{X,i}\) and \(e_{Z,i}\) unresolved; equivalently, the erased qubit is represented by the four possibilities \(I,X,Y,Z\).  Quantum-erasure decoding has been studied for surface codes, color codes, hypergraph product codes, and more general quantum LDPC codes using linear-time maximum-likelihood decoding, trimming, fast erasure decoding, belief propagation with guided decimation, cluster decomposition, degeneracy-aware BP, stabilizer-assisted inactivation, and quantum Maxwell erasure decoding \cite{DelfosseZemor2020,LeeMhallaSavin2020,ConnollyLondeLeverrierDelfosse2024,GokdumanYaoPfister2024,YaoGokdumanPfister2025,KuoOuyang2026,PechGokdumanYaoPfister2026,FreireLeRegentLeverrier2026}.  The deterministic density-evolution recursion associated with this model is introduced in \Cref{sec:uncoupled-joint}.

The decoding objective is to recover an estimate \((\widehat{\ve}_X,\widehat{\ve}_Z)\) from \((\bm{\sigma}_X,\bm{\sigma}_Z)\) and \(S\) with the same syndrome, such that the residual is CSS-stabilizer equivalent to the true error:
\[
      \widehat{\ve}_X+\ve_X\in\Row(H_X),\qquad
      \widehat{\ve}_Z+\ve_Z\in\Row(H_Z).
\]
For a hard-erasure decoder in these models, the basic state is whether the relevant visible component has been resolved or remains unknown.  Here and throughout the paper, a visible component means an unpunctured component that remains as a physical CSS coordinate, in contrast to the hidden auxiliary components that are punctured in the sparse representation.

The next definitions fix the random matrix ensembles used in the sparse representation.  We state both the uncoupled and tail-biting spatially coupled versions so that the finite ensembles are specified independently of any decoding dynamics.

\begin{definition}[\((j,k,M)\)-regular random matrix]\label{def:regular-random-matrix}
Let \(j,k,M\) be positive integers with \(M\geq k\) such that \(k\) divides \(jM\), and put \(m=jM/k\).  A \((j,k,M)\)-regular random matrix is a binary matrix
\[
      A\in\F_2^{m\times M}
\]
drawn uniformly from the set of matrices with every column of weight \(j\) and every row of weight \(k\).  Equivalently, it is the regular specialization of the standard socket ensemble \(\mathrm{LDPC}(\Lambda,P)\) in \cite[Def. 3.15]{RichardsonUrbanke2008}, with \(M\) variable nodes of degree \(j\) and \(m\) check nodes of degree \(k\), conditioned on having no parallel edges.
\end{definition}

This first ensemble is the uncoupled building block.  The next definition
records the corresponding tail-biting coupled block matrix.  It is the
exact-socket version of the \((\ell,r,L,w)\) spatially coupled ensemble in
\cite{KudekarRichardsonUrbanke2011}.

\begin{definition}[\((j,k,M,L,w)\)-spatially coupled regular random matrix]\label{def:sc-regular-random-matrix}
Let \(j,k,M,L,w\) be positive integers such that \(1\leq w<L\), \(k\mid jM\), and \(w\mid jM\), and put \(m=jM/k\).  Assume that the simple socket configurations specified below form a nonempty set.  A \((j,k,M,L,w)\)-spatially coupled regular random matrix is a tail-biting binary matrix
\[
      A^{\mathrm{SC}}\in\F_2^{Lm\times LM}
\]
with variable sections \(i\in\mathbb Z/L\mathbb Z\), each containing \(M\) columns, and check sections \(c\in\mathbb Z/L\mathbb Z\), each containing \(m\) rows.  In each variable section, uniformly at random partition the \(jM\) variable sockets into \(w\) groups of size \(jM/w\), indexed by offsets \(s=0,\ldots,w-1\).  Independently, in each check section, uniformly at random partition the \(km=jM\) check sockets into \(w\) groups of the same size.  For every section \(i\) and offset \(s\), match uniformly the \(s\)-th variable-socket group of section \(i\) to the \(s\)-th check-socket group of section \(i+s\) modulo \(L\), conditioned on simplicity.  The resulting matrix has column weight \(j\), row weight \(k\), and nonzero blocks only between variable section \(i\) and check sections \(i,\ldots,i+w-1\) modulo \(L\).
\end{definition}

The following ensemble builds the visible CSS pair from an MN/HA-type sparse representation with punctured auxiliary coordinates.  The extended matrices are sparse and define the local constraints of the representation.  The visible-coordinate matrices \(H_X,H_Z\), however, are generally dense after puncturing.

\begin{definition}[Nested regular ensemble through punctured sparse representations]\label{def:ensemble}
Fix integers \(j_Z,j_X,k\) with \(1\leq j_Z<j_X<k\).  Let \(N\geq k\) be a block parameter chosen so that the row counts below are integral.  Draw a \((j_Z,k,N)\)-regular random matrix
\[
      \AZ\in\F_2^{m_Z\times N},
      \qquad m_Z=\frac{j_Z}{k}N,
\]
in the sense of \Cref{def:regular-random-matrix}.  Independently draw a \((j_X-j_Z,k,N)\)-regular random matrix
\[
      \AD\in\F_2^{m_\Delta\times N},
      \qquad m_\Delta=\frac{j_X-j_Z}{k}N,
\]
and set
\[
      \AX=\begin{bmatrix}\AZ\\ \AD\end{bmatrix}.
\]
Then \(\Row(\AZ)\subseteq\Row(\AX)\), \(\AX\in\F_2^{m_X\times N}\), \(m_X=m_Z+m_\Delta=(j_X/k)N\), and each column of \(\AX\) has weight \(j_X\).  This is a regular degree profile with a prescribed row partition; \(\AX\) is not a uniform draw from the unpartitioned regular-matrix ensemble.  Independently of both A-blocks, draw a \((k,k,N)\)-regular random matrix
\[
      B\in\F_2^{N\times N}
\]
in the sense of \Cref{def:regular-random-matrix}.

The rightmost \(N\) coordinates are the visible variable \(\bm v\in\F_2^N\), and the left block consists of hidden variables that are punctured.  As in \cite{Kasai2026FiniteDegree}, define the extended sparse parity-check matrices
\begin{equation}
      H'_Z=
      \begin{bmatrix}
        \AZ & 0\\
        B   & I_N
      \end{bmatrix},
      \qquad
      H'_X=
      \begin{bmatrix}
        \AX^T & B^T
      \end{bmatrix}.
      \label{eq:extended-sparse-matrices}
\end{equation}
The left hidden variable is \(\bm u\in\F_2^N\) for \(H'_Z\) and \(\bm w\in\F_2^{m_X}\) for \(H'_X\).  The visible codes are obtained by puncturing the left hidden-variable part of the kernels:
\begin{equation}
\begin{aligned}
      \CZ
      &=\{\bm v\in\F_2^N:\exists\,\bm u\in\F_2^N,\
          H'_Z\begin{bmatrix}\bm u\\ \bm v\end{bmatrix}=\bm{0}\},\\
      \CX
      &=\{\bm v\in\F_2^N:\exists\,\bm w\in\F_2^{m_X},\
          H'_X\begin{bmatrix}\bm w\\ \bm v\end{bmatrix}=\bm{0}\}.
\end{aligned}
      \label{eq:punctured-visible-codes}
\end{equation}
Equivalently,
\[
      \CZ=B(\Ker \AZ),\qquad
      \CX=\{\bm v\in\F_2^N:B^T\bm v\in\Row(\AX)\}.
\]
\end{definition}

A tail-biting spatially coupled sparse representation is obtained by applying
\Cref{def:sc-regular-random-matrix} to each construction matrix.  The next
definition specifies only the finite sparse representation.  The seeded
boundary condition is a separate DE object defined in
\Cref{def:seeded-sc-de} and applied in \Cref{thm:joint-saturation}.

\begin{definition}[Spatially coupled punctured sparse representation]\label{def:sc-punctured-ensemble}
Fix positive integers \(j_Z,k_Z,j_\Delta,k_\Delta,k_B,M,L,w\).  Assume that
\[
      k_Z\mid j_ZM,\qquad k_\Delta\mid j_\Delta M,\qquad
      w\mid j_ZM,\qquad w\mid j_\Delta M,\qquad w\mid k_BM .
\]
Put
\[
      n=LM,\qquad
      m_Z=\frac{j_ZM}{k_Z},\qquad
      m_\Delta=\frac{j_\Delta M}{k_\Delta},\qquad
      m_X=m_Z+m_\Delta .
\]
Draw independently
\[
      \AZ^{\mathrm{SC}}\sim(j_Z,k_Z,M,L,w),\qquad
      \AD^{\mathrm{SC}}\sim(j_\Delta,k_\Delta,M,L,w),\qquad
      B^{\mathrm{SC}}\sim(k_B,k_B,M,L,w),
\]
in the sense of \Cref{def:sc-regular-random-matrix}, taking \(1\leq w<L\) and a section size for which all three simple ensembles are nonempty, and set
\[
      \AX^{\mathrm{SC}}
      =
      \begin{bmatrix}
        \AZ^{\mathrm{SC}}\\
        \AD^{\mathrm{SC}}
      \end{bmatrix}
      \in\F_2^{Lm_X\times n}.
\]
The spatially coupled extended sparse matrices are
\begin{equation}
      H_Z^{\prime\mathrm{SC}}
      =
      \begin{bmatrix}
        \AZ^{\mathrm{SC}} & 0\\
        B^{\mathrm{SC}}   & I_n
      \end{bmatrix}
      \in\F_2^{L(m_Z+M)\times 2n},
      \qquad
      H_X^{\prime\mathrm{SC}}
      =
      \begin{bmatrix}
        (\AX^{\mathrm{SC}})^T & (B^{\mathrm{SC}})^T
      \end{bmatrix}
      \in\F_2^{n\times L(m_X+M)} .
      \label{eq:sc-extended-sparse-matrices}
\end{equation}
The associated visible codes are obtained by puncturing the hidden coordinates:
\[
\begin{aligned}
      \CZ^{\mathrm{SC}}
      &=\{\bm v\in\F_2^n:\exists\,\bm u\in\F_2^n,\
          H_Z^{\prime\mathrm{SC}}
          \begin{bmatrix}\bm u\\ \bm v\end{bmatrix}=\bm0\},\\
      \CX^{\mathrm{SC}}
      &=\{\bm v\in\F_2^n:\exists\,\bm w\in\F_2^{Lm_X},\
          H_X^{\prime\mathrm{SC}}
          \begin{bmatrix}\bm w\\ \bm v\end{bmatrix}=\bm0\}.
\end{aligned}
\]
This is a tail-biting finite ensemble.  The deterministic seed used in
\Cref{thm:joint-saturation} is an additional DE boundary condition and is not
part of \eqref{eq:sc-extended-sparse-matrices}.  The degree notation used
in the DE sections is the special case \(k_Z=k_\Delta=k_B=k\) and
\(j_\Delta=j_X-j_Z\).
\end{definition}

The definition above fixes the graph ensemble before any decoding dynamics are
introduced.  The following table only summarizes the block sizes and degrees
so that the nested construction can be read without unpacking the whole
definition each time.

\begin{center}
\refstepcounter{table}\label{tab:construction-degrees}
\footnotesize
\textbf{Table~\thetable}\quad
Construction matrices and degrees.  Here \(m_Z=(j_Z/k)N\), \(m_\Delta=((j_X-j_Z)/k)N\), and \(m_X=(j_X/k)N\).
\par\smallskip
\begin{tabular}{@{}p{0.16\linewidth}p{0.22\linewidth}p{0.22\linewidth}p{0.30\linewidth}@{}}
\toprule
Matrix & Size & Degree & Comment \\
\midrule
\(\AZ\)
& \(m_Z\times N\)
& \((j_Z,k)\)-regular
& Gives the Z-side A-check \eqref{eq:z-a-check}.\\
\(\AD\)
& \(m_\Delta\times N\)
& \((j_X-j_Z,k)\)-regular
& Additional row block in \(\AX\).\\
\(\AX=[\AZ;\AD]\)
& \(m_X\times N\)
& \((j_X,k)\)-regular
& Appears as \(\AX^T\) in the X-side check \eqref{eq:x-side}.\\
\(B\)
& \(N\times N\)
& \((k,k)\)-regular
& Appears in the Z-side B-check \eqref{eq:z-b-check} and the X-side check \eqref{eq:x-side}.\\
\(H'_Z\)
& \((m_Z+N)\times 2N\)
& --
& Z-side extended sparse matrix.\\
\(H'_X\)
& \(N\times(m_X+N)\)
& --
& X-side extended sparse matrix.\\
\(H_X\)
& \(\rho_X\times N\)
& Generally dense
& Dense visible-coordinate row basis with \(\Row(H_X)=\CX^\perp\).\\
\(H_Z\)
& \(\rho_Z\times N\)
& Generally dense
& Dense visible-coordinate row basis with \(\Row(H_Z)=\CZ^\perp\).\\
\bottomrule
\end{tabular}
\end{center}

The next proposition checks that the visible codes obtained from the punctured sparse representation satisfy the CSS orthogonality condition.  The DE is written on the sparse representation, but this step fixes its relation to the usual dense visible-coordinate parity-check matrices.

\begin{proposition}[Dense parity-check matrices]\label{prop:dense-css-checks}
The visible codes of \Cref{def:ensemble} satisfy \(\CZ^\perp\subseteq\CX\).  Hence \(\CX,\CZ\) define a CSS pair.  The corresponding dense visible-coordinate parity-check matrices \(H_X,H_Z\) are any row-basis matrices satisfying
\begin{equation}
      \Row(H_X)=\CX^\perp,\qquad
      \Row(H_Z)=\CZ^\perp .
      \label{eq:dense-check-rowspaces}
\end{equation}
Since \(\Row(H_Z)=\CZ^\perp\subseteq\CX=(\Row(H_X))^\perp\), every row of \(H_Z\) is orthogonal to every row of \(H_X\), and hence
\[
      H_XH_Z^T=\bm{0}.
\]

More explicitly, if \(K_X\) is a row-basis matrix of \(\Ker \AX\), one may take
\[
      H_X=K_XB^T .
\]
One then replaces \(K_XB^T\) by a row basis of the same row space.  Also
\[
      \Row(H_Z)
      =
      \{\bm v\in\F_2^N:\exists\,\bm x\in\F_2^{m_Z},\
      \AZ^T\bm x+B^T\bm v=\bm{0}\},
\]
which is obtained by taking the kernel of \([\AZ^T\ B^T]\), splitting hidden and visible coordinates, and projecting to the visible block.  In general, these compressed matrices \(H_X,H_Z\) are dense.
\end{proposition}

\begin{proof}
We have
\[
      \CZ^\perp
      =\{\bm v\in\F_2^N:B^T\bm v\in(\Ker\AZ)^\perp=\Row(\AZ)\}.
\]
Since \(\Row(\AZ)\subseteq\Row(\AX)\), this gives \(\CZ^\perp\subseteq\{\bm v:B^T\bm v\in\Row(\AX)\}=\CX\).  This is the orthogonality stated above.  Moreover, \(\CX=\{\bm v:K_XB^T\bm v=\bm{0}\}\), so the row space of \(K_XB^T\) is the parity-check row space of \(\CX\).  The expression for \(H_Z\) is exactly the displayed expression for \(\CZ^\perp\).
\end{proof}

The bitmaps in
\Cref{fig:finite-paper-extended-checks,fig:finite-paper-compressed-checks}
reproduce the extended sparse matrices \(H'_Z,H'_X\) and the compressed
visible-coordinate matrices \(H_Z,H_X\) from the finite example in
\cite{Kasai2026FiniteDegree}.  They are included here to show the consequence of
\Cref{prop:dense-css-checks}: after puncturing the hidden variables, the actual
visible parity-check matrices need not retain the sparsity of the extended
matrices \(H'_Z,H'_X\) in \eqref{eq:extended-sparse-matrices}.  The example uses
the notation of \cite{Kasai2026FiniteDegree}, with
\((j_Z,k_Z,j_\Delta,k_\Delta,k)=(3,8,2,8,2)\), \(n=40\),
\(m_Z=15\), \(m_\Delta=10\), and \(m_X=25\).  The finite-instance
parameters are summarized in \Cref{tab:finite-example-parameters}.

\begin{table}[t]
\centering
\caption{Parameters of the finite CSS instance displayed in
\Cref{fig:finite-paper-extended-checks,fig:finite-paper-compressed-checks}.  The
design dimension is the value predicted from the sparse block sizes, while the
actual dimension is computed from the displayed visible row bases.}
\label{tab:finite-example-parameters}
\small
\begin{tabular}{@{}p{0.30\linewidth}p{0.26\linewidth}p{0.34\linewidth}@{}}
\toprule
Quantity & Value & Meaning \\
\midrule
Degree tuple
& \((j_Z,k_Z,j_\Delta,k_\Delta,k)=(3,8,2,8,2)\)
& \(A_Z\) is \((3,8)\)-regular, \(A_\Delta\) is \((2,8)\)-regular, and \(B\) is \((2,2)\)-regular.\\
Visible length
& \(n=40\)
& Number of physical visible coordinates.\\
Sparse block row counts
& \((m_Z,m_\Delta,m_X)=(15,10,25)\)
& Here \(m_X=m_Z+m_\Delta\).\\
Extended matrices
& \(H'_Z\in\F_2^{55\times80}\), \(H'_X\in\F_2^{40\times65}\)
& Matrices shown in \Cref{fig:finite-paper-extended-checks}.\\
Visible row bases
& \(H_Z\in\F_2^{16\times40}\), \(H_X\in\F_2^{15\times40}\)
& Matrices shown in \Cref{fig:finite-paper-compressed-checks}.\\
Design CSS dimension and rate
& \(K_Q^{\mathrm{des}}=10\), \(R_Q^{\mathrm{des}}=1/4\)
& Computed from \(n-(n-m_X)-m_Z=m_X-m_Z=25-15\).\\
Displayed CSS dimension and rate
& \(K_Q=9\), \(R_Q=9/40\)
& Computed as \(n-\operatorname{rank} H_X-\operatorname{rank} H_Z=40-15-16\).\\
\bottomrule
\end{tabular}
\end{table}

\begin{figure}[!htbp]
\centering
\makebox[0.92\textwidth][r]{\includegraphics[width=0.92\textwidth]{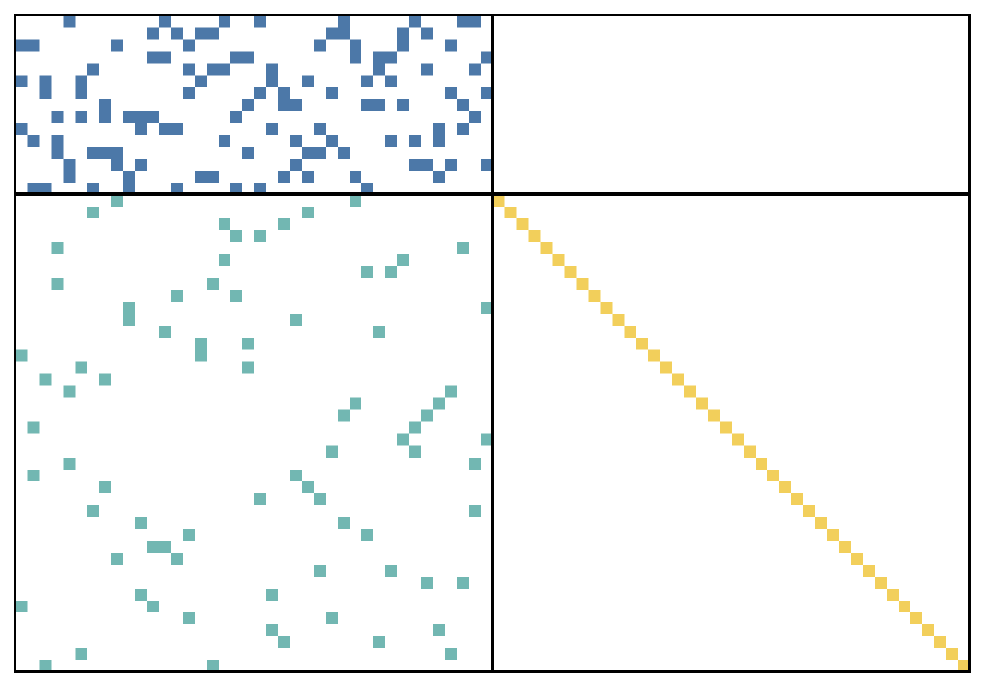}}
\par\medskip
\makebox[0.92\textwidth][r]{\includegraphics[width=0.7475\textwidth]{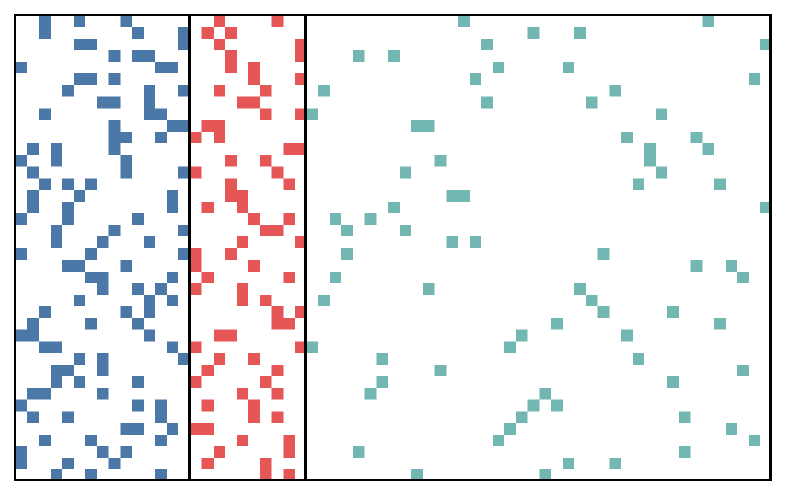}}
\caption{Extended sparse parity-check matrices reproduced from
\cite{Kasai2026FiniteDegree}.  The upper bitmap is the Z-side matrix
\(H'_Z\in\F_2^{55\times80}\), with blocks \(A_Z\), \(B\), and \(I_n\).  The lower
bitmap is the X-side matrix
\(H'_X=[A_X^T\ B^T]=[A_Z^T\ A_\Delta^T\ B^T]\in\F_2^{40\times65}\).  The two
bitmaps are shown with the same column scale and right alignment, so the
rightmost visible-coordinate blocks, both of width 40, have matching horizontal
extent.  Black lines indicate block boundaries.}
\label{fig:finite-paper-extended-checks}
\end{figure}

\begin{figure}[!htbp]
\centering
\includegraphics[width=0.66\textwidth]{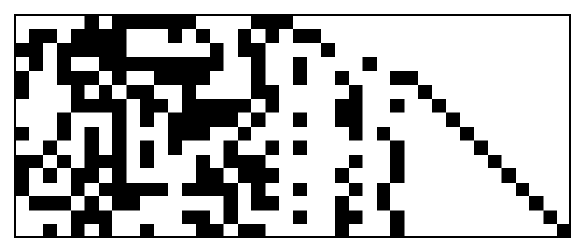}
\par\medskip
\includegraphics[width=0.66\textwidth]{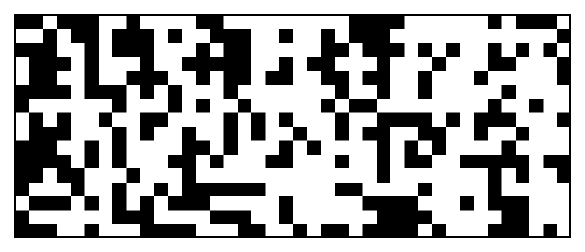}
\caption{Compressed visible-coordinate parity-check matrices reproduced from
\cite{Kasai2026FiniteDegree}.  The upper bitmap is the Z-side matrix
\(H_Z\in\F_2^{16\times40}\), obtained by projecting a basis of
\(\Ker[A_Z^T\ B^T]\) to the visible component.  The lower bitmap is the X-side
matrix \(H_X=K_XB^T\in\F_2^{15\times40}\), where \(K_X\) is a basis matrix of
\(\Ker A_X\).  No final
reduced row echelon form is applied in the displayed representatives.  The two
bitmaps are shown with the same horizontal scale because both matrices have 40
visible-coordinate columns.}
\label{fig:finite-paper-compressed-checks}
\end{figure}

\Cref{fig:sc-extended-sparse-checks} shows one small tail-biting spatially
coupled realization of \eqref{eq:sc-extended-sparse-matrices}.  The example
uses \(L=20\), \(w=2\), \(M=8\),
\[
      \AZ^{\mathrm{SC}}\sim(3,8,8,20,2),\qquad
      \AD^{\mathrm{SC}}\sim(2,8,8,20,2),\qquad
      B^{\mathrm{SC}}\sim(2,2,8,20,2).
\]
The two bitmaps display only the sparse extended matrices; the dense visible
matrices obtained after puncturing are different objects.

\begin{figure}[!htbp]
\centering
\makebox[0.96\textwidth][r]{\includegraphics[width=0.96\textwidth]{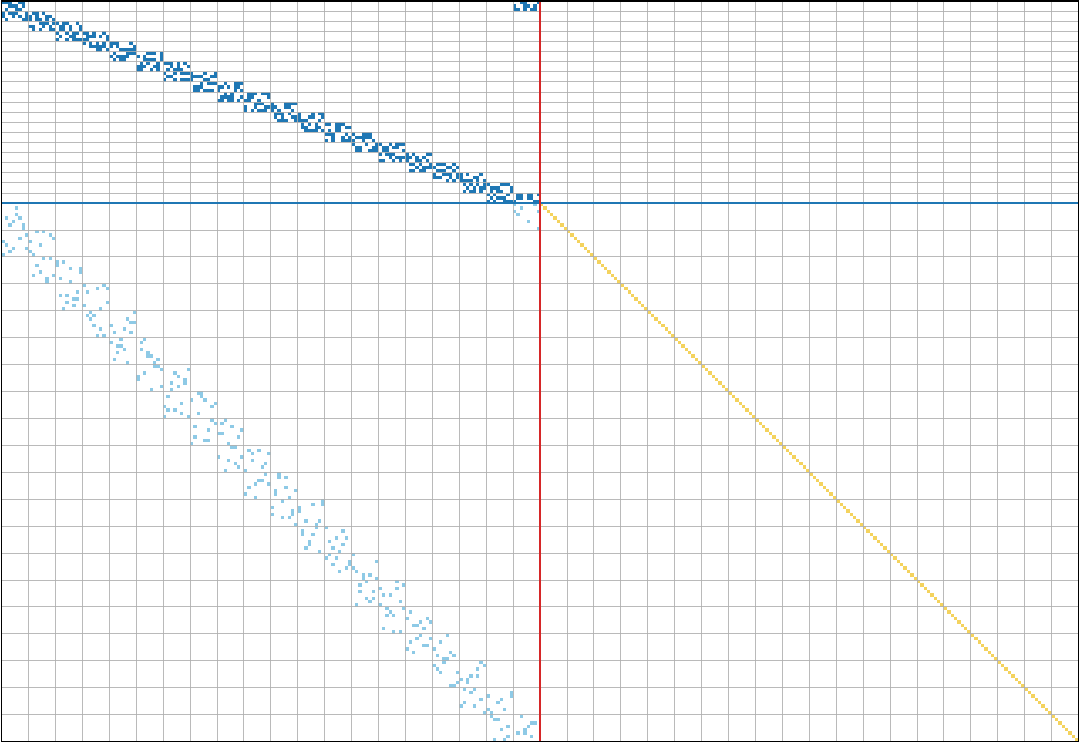}}
\par\medskip
\makebox[0.96\textwidth][r]{\includegraphics[width=0.78\textwidth]{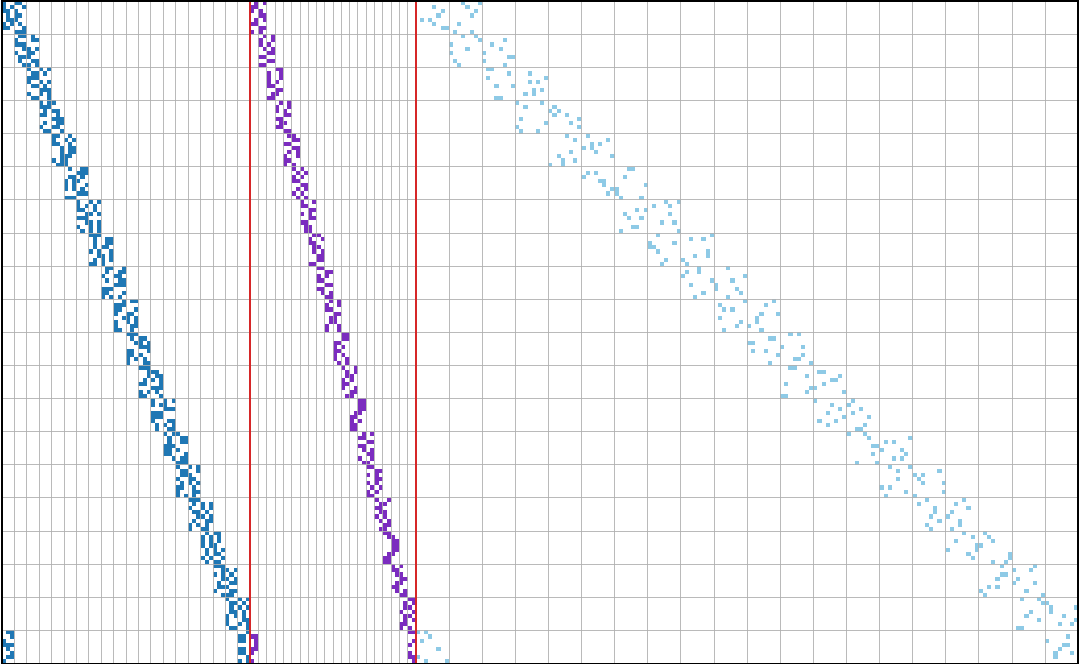}}
\caption{A tail-biting spatially coupled sparse-matrix realization.  The upper
bitmap is \(H_Z^{\prime\mathrm{SC}}\), with \(A_Z^{\mathrm{SC}}\) in blue,
\(B^{\mathrm{SC}}\) in light blue, and \(I_n\) in yellow.  The lower bitmap is
\(H_X^{\prime\mathrm{SC}}=[(\AX^{\mathrm{SC}})^T\ (B^{\mathrm{SC}})^T]\), with
\((A_Z^{\mathrm{SC}})^T\) in blue, \((A_\Delta^{\mathrm{SC}})^T\) in purple,
and \((B^{\mathrm{SC}})^T\) in light blue.  The two bitmaps are right-aligned
and scaled so that the visible, unpunctured \(n\)-coordinate blocks have the
same physical width.  Thin gray lines mark section boundaries; thick lines mark
block boundaries.}
\label{fig:sc-extended-sparse-checks}
\end{figure}

For the same realization, \Cref{fig:sc-dense-visible-checks} displays dense
visible-coordinate row bases obtained after puncturing the hidden coordinates.
The Z-side matrix is obtained by projecting a basis of
\(\Ker[(\AZ^{\mathrm{SC}})^T\ (B^{\mathrm{SC}})^T]\) to the visible block, and
the X-side matrix is \(K_X(B^{\mathrm{SC}})^T\), where \(K_X\) is a basis matrix
of \(\Ker \AX^{\mathrm{SC}}\).  This finite example illustrates that the
sparse matrices in \Cref{fig:sc-extended-sparse-checks} represent a generally
dense visible-coordinate CSS pair.

\begin{figure}[!htbp]
\centering
\includegraphics[width=0.82\textwidth]{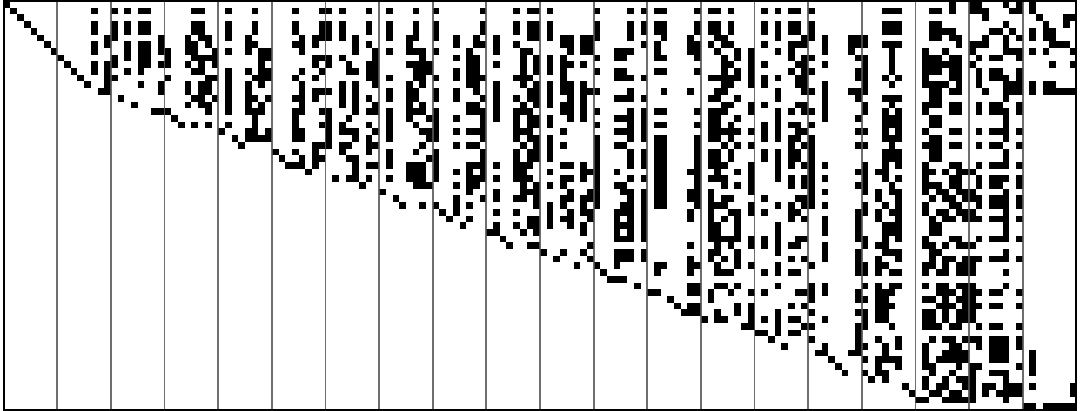}
\par\medskip
\includegraphics[width=0.82\textwidth]{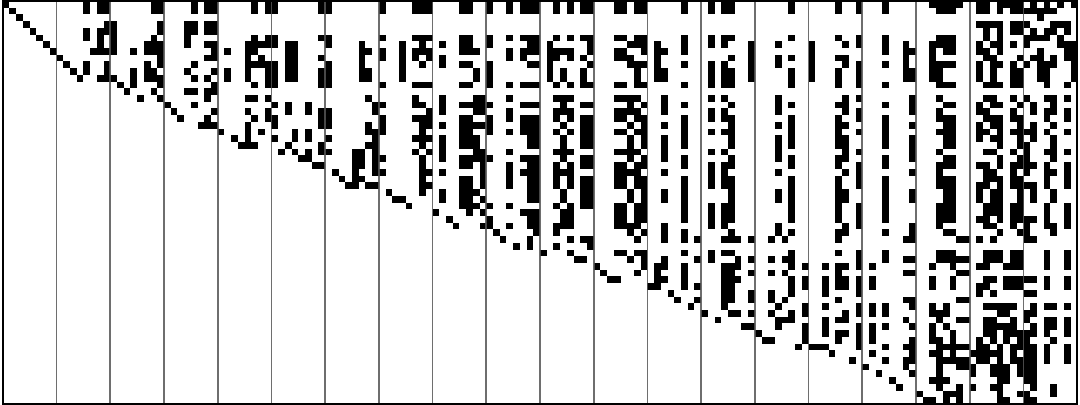}
\caption{Dense visible-coordinate parity-check row bases corresponding to the
tail-biting realization in \Cref{fig:sc-extended-sparse-checks}.  The upper
bitmap is the Z-side matrix \(H_Z\in\F_2^{61\times160}\).  The lower bitmap is
the X-side matrix \(H_X\in\F_2^{60\times160}\).  Both matrices are displayed on
the same horizontal scale because they have the same \(n=160\) visible
coordinates.}
\label{fig:sc-dense-visible-checks}
\end{figure}
\FloatBarrier

Choose any syndrome representatives \(\bm r_X,\bm r_Z\in\F_2^N\) satisfying
\[
 H_Z\bm r_X=\bm\sigma_Z,\qquad H_X\bm r_Z=\bm\sigma_X,
\]
and set \(\vt_Z=\bm r_X\), \(\vt_X=B^T\bm r_Z\).  Such representatives exist for every measured syndrome and can be chosen from the syndrome alone, without knowing the error.  The identities \(\Ker H_Z=\CZ\) and \(\Ker H_X=\CX\) then give the following equivalent sparse affine systems.  The Z-side A-check \eqref{eq:z-a-check} is
\begin{equation}
      \AZ\vf_Z=\bm{0}                                  \label{eq:z-a-check}
\end{equation}
and the Z-side B-check \eqref{eq:z-b-check} is
\begin{equation}
      \ve_X=\vt_Z+B\vf_Z.                               \label{eq:z-b-check}
\end{equation}
This is the sparse representation of the dense syndrome equation involving \(H_Z\ve_X\).  The X-side sparse equations use the stacked matrix and transpose form
\begin{equation}
      B^T\ve_Z=\vt_X+\AX^T\vg_X.                       \label{eq:x-side}
\end{equation}
This is the sparse representation of the dense syndrome equation involving \(H_X\ve_Z\).  We call \eqref{eq:z-a-check} a Z-side A-check, \eqref{eq:z-b-check} a Z-side B-check, and the constraint in \eqref{eq:x-side} an X-side check.  The constant offsets \(\vt_Z,\vt_X\) change only the recovered bit values; they do not change whether an erasure message is known or unresolved.

The next proposition computes the design rate of the sparse representation and
the hashing-bound channel parameter that is compared with the potential
threshold in \Cref{sec:threshold-saturation}.  A
quantum erasure leaves two unknown binary Pauli degrees of freedom per erased
coordinate.  The calculation is the degree-of-freedom form of the erasure
hashing/capacity expression in \cite{BennettErasure1997}.

For the actual finite code, put
\[
 \nu_Z=\dim(\Ker B\cap\Ker\AZ),\qquad
 \nu_X=\dim(\Ker B\cap\Ker\AX).
\]
The visible dimensions are \(\dim\CZ=N-\operatorname{rank}\AZ-\nu_Z\) and
\(\dim\CX=\operatorname{rank}\AX+\nu_X\), so
\begin{equation}
 K_Q=\operatorname{rank}\AX-\operatorname{rank}\AZ+\nu_X-\nu_Z
 \leq m_X-m_Z.
 \label{eq:actual-css-dimension}
\end{equation}
Indeed, kernel nesting gives \(\nu_X\leq\nu_Z\), and appending the
\(m_\Delta\) rows increases rank by at most \(m_\Delta\).
This is the rank/overlap identity of \cite[Proposition 2.4]{Kasai2026FiniteDegree}.
The design dimension discards these rank and overlap corrections.
Consequently the hashing parameter below is defined from the design rate;
identifying it with an actual asymptotic coding rate requires control of
those corrections and of the physical seed cost.

\begin{proposition}[Design rate and hashing-bound parameter]\label{prop:rate}
For the finite-degree ensemble of \Cref{def:ensemble} with common check degree
\(k\), the constituent design rates are
\[
      R_Z^{\mathrm{des}}=1-\frac{j_Z}{k},
      \qquad
      R_X^{\mathrm{des}}=\frac{j_X}{k},
\]
and the CSS design rate is
\[
      R_Q^{\mathrm{des}}=\frac{j_X-j_Z}{k}.
\]
The X/Z equal-rate specialization is the case
\(R_Z^{\mathrm{des}}=R_X^{\mathrm{des}}\), equivalently \(j_X=k-j_Z\).  It has
\[
      R_Q^{\mathrm{des}}=1-\frac{2j_Z}{k}.
\]
Equivalently, for a target design rate \(0<\rho<1\), any \(k\) such that
\[
      j_Z=\frac{1-\rho}{2}k,\qquad
      j_X=\frac{1+\rho}{2}k
\]
are integers gives an X/Z equal-rate degree triple with
\(R_Q^{\mathrm{des}}=\rho\).  A purely integral parametrization is
\[
      (j_Z,j_X,k)=(j,j+\lambda,2j+\lambda),
      \qquad
      R_Q^{\mathrm{des}}=\frac{\lambda}{2j+\lambda},
      \qquad j\geq2,\ \lambda\geq1 .
\]
For the quantum erasure channel, an erased coordinate leaves two unknown binary Pauli components, and the hashing-bound channel parameter is
\[
      \epsilon_{\mathrm{hash}}
        =\frac{1-R_Q^{\mathrm{des}}}{2}.
\]
In the design-rate-\(1/3\) example \((j_Z,j_X,k)=(4,8,12)\),
\[
      \epsilon_{\mathrm{hash}}=1/3 .
\]
\end{proposition}

\begin{proof}
Formally counting independent rows and no dimension loss under puncturing gives \(R_Z^{\mathrm{des}}=1-j_Z/k\) from the \(N\) hidden variables and \(m_Z=(j_Z/k)N\) A-constraints.  The corresponding X-side count gives design dimension \(m_Z+m_\Delta=(j_X/k)N\), hence \(R_X^{\mathrm{des}}=j_X/k\).  The CSS design rate is \(R_Z^{\mathrm{des}}+R_X^{\mathrm{des}}-1=(j_X-j_Z)/k\).
For the quantum erasure channel, each erased coordinate contributes two unresolved binary degrees of freedom, so the conditional degree-of-freedom density is \(2\epsilon\).  The hashing-bound condition is \(2\epsilon=1-R_Q^{\mathrm{des}}\).
\end{proof}

\section{Density Evolution for CSS Erasure Decoding}\label{sec:uncoupled-joint}

The ensembles in \Cref{sec:model} are finite random sparse-graph ensembles for
each block length.  The following standard density-evolution limit gives the
asymptotic meaning of the deterministic recursions used in this section.

\begin{theorem}[Density-evolution limit for local erasure decoding]\label{thm:de-limit}
Fix the degrees of the sparse CSS ensemble in \Cref{def:ensemble}, an erasure
probability \(\epsilon\), and a finite number \(\ell\) of BP iterations.  As
\(N\to\infty\), the depth-\(\ell\) computation neighborhood of a uniformly
chosen edge of each message type is tree-like with probability tending to one.
Consequently, the expected message-erasure probability after \(\ell\) iterations
converges to the value obtained from the corresponding tree recursion
\cite{RichardsonUrbanke2008}.  The same local limit applies sectionwise to the
tail-biting coupled ensemble of \Cref{def:sc-punctured-ensemble} when \(L,w,\ell\) are
fixed and \(M\to\infty\).  Block and section sizes range over admissible values for the stated simple ensembles.  The computation neighborhood is taken within one constituent graph; independence between the two constituent decoders is not required.  The uncoupled equations and the sectionwise equations in \Cref{sec:sectionwise-de} give the corresponding marginal limits.  The additional message seed is imposed separately.
\end{theorem}

This theorem justifies replacing a large finite computation neighborhood by a
tree recursion for any fixed number of iterations.  The rest of the section
derives that tree recursion explicitly for the five message types of
\Cref{def:ensemble}.

The next theorem gives the concrete recursion for the uncoupled finite-degree
case with fixed Z-side, X-side, and check degrees.  It identifies the
deterministic recursion whose fixed points are analyzed in
\Cref{sec:potential,sec:constituent-positivity}; it is not a finite-length
concentration or block-error theorem for the random CSS ensembles.

Let \(a,b\) be the erasure probabilities of messages from a Z-side punctured variable to the Z-side A-check \eqref{eq:z-a-check} and to the Z-side B-check \eqref{eq:z-b-check}, respectively.  Let \(c\) be the erasure probability of the Z-side visible-component message to the Z-side B-check \eqref{eq:z-b-check}.  Let \(d,e\) be the erasure probabilities of X-side auxiliary-variable and X-side visible-component messages to the X-side check \eqref{eq:x-side}, respectively.  The reverse check-to-variable erasure probabilities are
\begin{align}
    \hat{a} &=1-(1-a)^{k-1},                                      \label{eq:joint-hat-a}\\
    \hat{b} &=1-(1-c)(1-b)^{k-1},                                  \label{eq:joint-hat-b}\\
    \hat{c} &=1-(1-b)^k,                                           \label{eq:joint-hat-c}\\
    \hat{d} &=1-(1-d)^{j_X-1}(1-e)^k,                              \label{eq:joint-hat-d}\\
    \hat{e} &=1-(1-d)^{j_X}(1-e)^{k-1}.                            \label{eq:joint-hat-e}
\end{align}

\begin{theorem}[Hard-erasure density evolution]\label{thm:joint-uncoupled-de}
For CSS erasure decoding over the quantum erasure channel, the abstract local ternary erasure rule is described by
\begin{equation}
\begin{aligned}
    a^+ &=\hat{a}^{j_Z-1}\hat{b}^k,\\
    b^+ &=\hat{a}^{j_Z}\hat{b}^{k-1},\\
    c^+ &=\epsilon,\\
    d^+ &=\hat{d}^{k-1},\\
    e^+ &=\epsilon\,\hat{e}^{k-1}.
\end{aligned}                                               \label{eq:joint-de}
\end{equation}
The probabilities that the individual visible Pauli components remain unresolved are
\begin{equation}
      r_Z=\epsilon\,\hat{c},\qquad
      r_X=\epsilon\,\hat{e}^k.                                \label{eq:joint-residual}
\end{equation}
These are individual-component erasure probabilities.  Resolving all components with correct messages suffices for recovery, whereas logical recovery can also be possible with unresolved components related by stabilizers.
\end{theorem}

\begin{proof}
The check-to-variable updates \eqref{eq:joint-hat-a}--\eqref{eq:joint-hat-e} are the usual erasure-check updates for \eqref{eq:z-a-check}, \eqref{eq:z-b-check}, and \eqref{eq:x-side}.  A Z-side punctured variable sends an erasure to an A-check iff all other \(j_Z-1\) A-check messages and all \(k\) B-check messages are erased, giving \(a^+\).  It sends an erasure to a B-check iff all \(j_Z\) A-check messages and the other \(k-1\) B-check messages are erased, giving \(b^+\).  Information about one Pauli component does not determine the other component.  Hence the Z-side visible-component message to \eqref{eq:z-b-check} remains erased exactly when the channel erases the coordinate, giving \(c^+=\epsilon\).  The X-side auxiliary-variable update gives \(d^+=\hat d^{k-1}\).  The X-side visible-component message remains erased when the channel erases the coordinate and the other \(k-1\) X-side check messages are erased, giving \(e^+=\epsilon\hat e^{k-1}\).  The residual formulas require one additional incoming check message, giving \eqref{eq:joint-residual}.
\end{proof}

The recursion is closed because the hard-erasure rule tracks only whether each
message is resolved.  Its Z-side and X-side parts will be analyzed by separate
potentials, so we now record the two constituent recursions read off from
\eqref{eq:joint-de}.  The Z-side equations use
\eqref{eq:joint-hat-a}--\eqref{eq:joint-hat-c}:
\begin{equation}
\begin{aligned}
    a^+ &=\hat{a}^{j_Z-1}\hat{b}^k,\\
    b^+ &=\hat{a}^{j_Z}\hat{b}^{k-1},\\
    c^+ &=\epsilon,
\end{aligned}                                               \label{eq:ic-z-de}
\end{equation}
and the Z-side residual is
\begin{equation}
      r_Z=\epsilon\,\hat c .                   \label{eq:ic-z-residual}
\end{equation}
The X-side equations use
\eqref{eq:joint-hat-d}--\eqref{eq:joint-hat-e}:
\begin{equation}
\begin{aligned}
    d^+ &=\hat{d}^{k-1},\\
    e^+ &=\epsilon\,\hat{e}^{k-1},
\end{aligned}                                               \label{eq:ic-x-de}
\end{equation}
and the X-side residual is
\begin{equation}
      r_X=\epsilon\,\hat e^k .                 \label{eq:ic-x-residual}
\end{equation}

\subsection{Sectionwise DE of the sparse matrices}\label{sec:sectionwise-de}
For a scalar profile $z$, define the modulo-$L$ averages
\begin{equation}
 (\mathcal M_-z)_i=\frac1w\sum_{r=0}^{w-1}z_{i-r},\qquad
 (\mathcal M_+z)_i=\frac1w\sum_{r=0}^{w-1}z_{i+r}.
 \label{eq:section-averages}
\end{equation}
Let $\epsilon_i$ be the visible erasure probability in section $i$.
The random socket partitions in \Cref{def:sc-regular-random-matrix} give, in the section-size limit, the Z-side check updates
\begin{align}
 \hat a_i&=1-[1-(\mathcal M_-a)_i]^{k-1},\notag\\
 \hat b_i&=1-(1-\epsilon_i)[1-(\mathcal M_-b)_i]^{k-1},\label{eq:native-z-checks}\\
 \hat c_i&=1-[1-(\mathcal M_-b)_i]^k.\notag
\end{align}
The variable updates and visible residual are
\begin{align}
 a_i^+&=(\mathcal M_+\hat a)_i^{j_Z-1}(\mathcal M_+\hat b)_i^k,\notag\\
 b_i^+&=(\mathcal M_+\hat a)_i^{j_Z}(\mathcal M_+\hat b)_i^{k-1},\qquad c_i^+=\epsilon_i,\label{eq:native-z-variables}\\
 r_{Z,i}&=\epsilon_i\hat c_i.\notag
\end{align}
The factor $\epsilon_i$ in \eqref{eq:native-z-checks} is local: the visible variable is attached to the B-check through $I_n$, so this edge is not window-averaged.

For the X-side graph, transposition reverses the offsets.  Auxiliary and visible variables indexed by an original row section receive messages from original column sections.  Thus
\begin{align}
 \hat d_i&=1-[1-(\mathcal M_+d)_i]^{j_X-1}[1-(\mathcal M_+e)_i]^k,\notag\\
 \hat e_i&=1-[1-(\mathcal M_+d)_i]^{j_X}[1-(\mathcal M_+e)_i]^{k-1},\label{eq:native-x-checks}\\
 d_i^+&=(\mathcal M_-\hat d)_i^{k-1},\qquad
 e_i^+=\epsilon_i(\mathcal M_-\hat e)_i^{k-1},\notag\\
 r_{X,i}&=\epsilon_i(\mathcal M_-\hat e)_i^k.\notag
\end{align}
The two auxiliary row blocks have the same degree $k$ and the same limiting offset distribution, so their marginal messages have a common value $d_i$ under the common initialization.  Reflecting section indices converts the X-side orientation into the usual $\mathcal M_+$ after $\mathcal M_-$ convention.

These equations follow by first conditioning on the check section and then averaging the independently selected socket offsets.  At a variable, the message on the outgoing edge is excluded.  A section-constant profile reduces them to \eqref{eq:joint-de}.  In contrast, applying a common vector average to $(a,b,c)$ averages $c$ as well, and defines a different recursion near a nonuniform channel boundary.  We distinguish that abstract recursion in the numerical comparison below.

An \emph{ideal message seed} clamps $a_i,b_i,d_i,e_i$ to zero on a specified interval; when $\epsilon_i=0$ there, it also gives $c_i=0$.  This is stronger than setting the visible channel to zero.  Indeed, $a_i=b_i=1$ at every section implies $\hat a_i=\hat b_i=1$ and $a_i^+=b_i^+=1$ even if selected $\epsilon_i$ vanish.  Thus visible erasure removal alone does not initialize the required auxiliary messages.

\section{Vector-Potential Preliminaries}\label{sec:potential-preliminaries}

The shift argument used here is a fixed-channel argument.  Its scalar parameter is held fixed throughout the proof.  This permits the positive diagonal weight matrix and the primitives to depend on that parameter.

\begin{definition}[Fixed-channel potential system]\label{def:admissible-vector-system}
Fix $\epsilon$.  On $[0,1]^d$, let $f$ and $g$ be twice continuously differentiable, coordinatewise nondecreasing maps into $[0,1]^d$, with $f(0)=g(0)=0$.  Suppose that a positive diagonal matrix $D$ and scalar functions $F,G$ satisfy
\[
 \nabla F(y)=Df(y),\qquad \nabla G(x)=Dg(x),\qquad F(0)=G(0)=0.
\]
Dependence on the fixed $\epsilon$ is suppressed in this definition.  No claim that $f(x;0)=0$ is imposed on a punctured constituent.  Since the gradients are nonnegative, $F$ and $G$ are nonnegative on the state cube.
\end{definition}

The primitive identities imply
\[
 \nabla U(x;\epsilon)=g'(x)^{\mathsf T}D[x-f(g(x);\epsilon)].
\]
Every DE fixed point is consequently a stationary point of $U$.  The converse would require additional conditions on $g'$ and is not used.

\begin{definition}[Vector potential for an uncoupled DE system]\label{def:general-vector-potential}
For a fixed-channel potential system, with \(D,F,G\) as in
\Cref{def:admissible-vector-system}, the associated vector potential is
\begin{equation}
      U(\bm{x};\epsilon)
        =\bm{g}(\bm{x})^T D\bm{x}
         -G(\bm{x})-F(\bm{g}(\bm{x});\epsilon).              \label{eq:general-vector-potential}
\end{equation}
With the convention \(\inf\emptyset=+\infty\), define
\begin{equation}
      \mathcal F^\star(\epsilon)
      =
      \{\bm{x}\in[0,1]^d\setminus\{\bm{0}\}:
        \bm{x}=\bm{f}(\bm{g}(\bm{x});\epsilon)\},
      \qquad
      \Delta E(\epsilon)
      =
      \inf_{\bm{x}\in\mathcal F^\star(\epsilon)}
      U(\bm{x};\epsilon).                                    \label{eq:general-energy-gap}
\end{equation}
The potential threshold is
\begin{equation}
      \epsilon_{\mathrm{pot}}
      =
      \sup\{\epsilon_0\in[0,1]:
        \Delta E(\epsilon)>0
        \text{ for every }0\leq\epsilon<\epsilon_0\}.          \label{eq:general-potential-threshold}
\end{equation}
\end{definition}

This function is useful because fixed points of the corresponding DE recursion
are critical configurations of the potential, and the coupled-vector potential
argument measures their cost by \(U\).  The energy gap
\(\Delta E(\epsilon)\) in \eqref{eq:general-energy-gap} is the quantity that
controls threshold saturation after spatial coupling.

The next definition fixes the seeded spatially coupled recursion used in the
threshold statement.  It is written for a general vector DE system, so the
same formula applies to any constituent system satisfying
\Cref{def:admissible-vector-system}.

\begin{definition}[Seeded spatially coupled DE]\label{def:seeded-sc-de}
Let the state space be \([0,1]^d\), ordered coordinatewise, and let
\(\bm{0}\) be the successful state.  Fix a coupling length \(L\), a coupling
width \(w\) with \(1\leq w<L\), and a seed interval
\(\mathcal S\subset\mathbb Z/L\mathbb Z\) containing at least $w$ consecutive sections.  Indices below are taken modulo
\(L\).  For a profile
\(\bm{X}^{(\ell)}=(\bm{x}_0^{(\ell)},\ldots,\bm{x}_{L-1}^{(\ell)})\), define
\begin{align}
      \bar{\bm{x}}_{c}^{(\ell)}
        &=\frac1w\sum_{r=0}^{w-1}\bm{x}_{c-r}^{(\ell)},       \label{eq:general-sc-x-average}\\
      \bm{y}_{c}^{(\ell)}
        &=\bm{g}(\bar{\bm{x}}_{c}^{(\ell)}),                  \label{eq:general-sc-g-update}\\
      \bar{\bm{y}}_{i}^{(\ell)}
        &=\frac1w\sum_{r=0}^{w-1}\bm{y}_{i+r}^{(\ell)},       \label{eq:general-sc-y-average}\\
      \bm{x}_{i}^{(\ell+1)}
        &=
        \begin{cases}
          \bm{0}, & i\in\mathcal S,\\
          \bm{f}(\bar{\bm{y}}_{i}^{(\ell)};\epsilon),
             & i\notin\mathcal S .
        \end{cases}                                          \label{eq:general-seeded-sc-de}
\end{align}
The seeded DE is initialized by
\[
      \bm{x}_{i}^{(0)}
        =
        \begin{cases}
          \bm{0}, & i\in\mathcal S,\\
          \bm{1}, & i\notin\mathcal S,
        \end{cases}
\]
where \(\bm{1}\) denotes the largest state in \([0,1]^d\).  The recursion is
successful in the non-seeded region if
\[
      \lim_{\ell\to\infty}\bm{x}_{i}^{(\ell)}=\bm{0}
      \qquad\text{for every } i\notin\mathcal S .
\]
\end{definition}

\begin{theorem}[Fixed-channel seeded convergence]\label{thm:standard-vector-saturation}
Fix a system satisfying \Cref{def:admissible-vector-system} and suppose its nonzero fixed points have energy gap $\Delta E>0$.  Then the recursion in \Cref{def:seeded-sc-de} converges to zero for all unseeded sections when $w$ is sufficiently large.  The required width may depend on this fixed system and its energy gap, but not on the unseeded chain length.
\end{theorem}

\begin{proof}
We give the boundary and norm details of the fixed-channel shift argument.
If every section is seeded there is nothing to prove. Otherwise let $n\geq1$
be the number of unseeded sections. Cutting the ring inside the seed, whose
width is at least $w$, gives $n$ consecutive variable sections with zero
exterior variables. No length-$w$ check window meets unseeded variables on
both sides of the cut. For any integer $r\geq1$, define $A_r$ to be the
$r\times(r+w-1)$ matrix
\[
 (A_r)_{j,i}=\frac1w\,\mathbf1\{j\leq i\leq j+w-1\}.
\]
It acts on each message coordinate separately. In particular,
$\|A_r\|_\infty=1$ and $\|A_r^{\mathsf T}\|_\infty\leq1$.
For the actual unseeded variable profile $x$, put $z=A_n^{\mathsf T}x$.
Its update is
\[
 z^+=T_n(z),\qquad T_r(Z):=A_r^{\mathsf T}f(A_rg(Z)).
\]

Here is an explicit finite one-sided comparison, valid for every $n,w$.
Set $m=n+w-1$ and $N=m+w-1$. On $N$ averaged sections define
\[
 [Q(Z)]_i=
 \begin{cases}
 [T_m(Z)]_i,&1\leq i\leq m,\\
 [T_m(Z)]_m,&m<i\leq N.
 \end{cases}
\]
Initialize $Z^{(0)}$ to the all-one profile. The map $Q$ is monotone and
preserves the cube, so these iterates decrease to a fixed profile $Z$.
They also remain nondecreasing in the section index. Indeed, for any such
profile, the vectors $v_j=f((A_mg(Z))_j)$ are nondecreasing for $1\leq j\leq m$.
With $v_j=0$ for $j\leq0$, the updating sections satisfy
\[
 [Q(Z)]_i=\frac1w\sum_{r=0}^{w-1}v_{i-r},\qquad 1\leq i\leq m,
\]
which is nondecreasing; the remaining sections copy its final value.
Padding the original $n+w-1$ averaged sections with zeros to length $N$
shows by induction that the $Q$ iterates dominate the actual averaged
profile. On every original section, which has index at most $m$, the
comparison retains all original $f$-sites and adds only nonnegative terms.
Thus it suffices to prove that $Z=0$.

Write $z_*=Z_m$, so that $Z_i=z_*$ for $m\leq i\leq N$.
If $Z\ne0$, spatial monotonicity gives $z_*\ne0$. Since every window in
the fixed-point equation at section $m$ has input at most $z_*$,
\[
 z_*=[T_m(Z)]_m\leq f(g(z_*)).
\]
Uncoupled iteration from $z_*$ consequently increases to a nonzero fixed
point $z_\infty$. Its potential does not increase along this ordered
iteration. To check this directly, if $y=f(g(x))\geq x$, every point
$q=x+t(y-x)$ of the connecting segment satisfies
$q\leq y\leq f(g(q))$. Hence
\[
 \nabla U(q)^{\mathsf T}(y-x)
 =[q-f(g(q))]^{\mathsf T}Dg'(q)(y-x)\leq0,
\]
because the Jacobian of the monotone map $g$ has nonnegative entries.
Continuity therefore gives
\[
 U(z_*)\geq U(z_\infty)\geq\Delta E.
\]
If the uncoupled nonzero fixed-point set is empty, the existence of
$z_\infty$ is already a contradiction.

For the remaining case $0<\Delta E<\infty$, use the potential
\[
 \mathcal U_m(Z)=\sum_{i=1}^{N}
 [g(Z_i)^{\mathsf T}DZ_i-G(Z_i)]
 -\sum_{j=1}^{m}F((A_mg(Z))_j).
\]
Its section gradient is
$g'(Z_i)^{\mathsf T}D[Z_i-T_m(Z)_i]$.
Define the shift by $(SZ)_1=0$ and $(SZ)_i=Z_{i-1}$ for $i\geq2$.
The local part of the potential difference telescopes to
$-[g(z_*)^{\mathsf T}Dz_*-G(z_*)]$.
For the $F$ part, all windows except the first shift by one site; the
last unshifted window, indexed by $m$, lies entirely in the plateau.
Since $F\geq0$,
\[
 \mathcal U_m(SZ)-\mathcal U_m(Z)\leq-U(z_*)\leq-\Delta E.
\]
The linear Taylor term is exactly zero. The gradient vanishes for
$i\leq m$ by the $Q$ fixed-point equations, whereas $(SZ-Z)_i=0$ for
$i>m$ by the plateau condition.

The shift has the two bounds
\[
 \|\operatorname{vec}(SZ-Z)\|_\infty\leq\frac1w,
 \qquad
 \|\operatorname{vec}(SZ-Z)\|_1=\|z_*\|_1\leq d.
\]
For the first, the fixed equations on $i\leq m$ give
$Z_i-Z_{i-1}=(v_i-v_{i-w})/w$, where $Z_0=0$ and $v_j=0$ for $j\leq0$;
each coordinate of this difference is at most $1/w$. On the plateau
the differences vanish. The second bound is componentwise telescoping.

For clarity, the Hessian bound can be made independent of $m$ and $w$
without a sign assumption on the Hessian. Use the column-Jacobian
convention and set
\[
 M_f=\sup_y\|f'(y)\|_\infty,\qquad
 M_g=\sup_x\max\{\|g'(x)\|_1,\|g'(x)\|_\infty\},
\]
\[
 H_g=\sup_x\max_{1\leq b\leq d}
       \sum_{a,c=1}^{d}
       \left|\frac{\partial^2g_a(x)}{\partial x_b\partial x_c}\right|,
 \qquad
 K=\|D\|_\infty\bigl(H_g+M_g+M_fM_g^2\bigr).
\]
These constants are finite on the compact state cube.
Stack coordinates section by section, and let $J_g$ denote the block
diagonal Jacobian of sectionwise $g$.
Differentiating the profile gradient gives
\[
 \nabla^2\mathcal U_m=R+J_g^{\mathsf T}(I_N\otimes D)(I-T_m'),
\]
where the $i$-th diagonal block of $R$ has $(b,c)$ entry
\[
 \sum_{a=1}^{d}D_{aa}[Z_{i,a}-T_m(Z)_{i,a}]
        \frac{\partial^2g_a(Z_i)}{\partial x_b\partial x_c}.
\]
Because both $Z$ and $T_m(Z)$ lie in the cube,
$\|R\|_\infty\leq\|D\|_\infty H_g$.
The averaging-matrix bounds give
$\|T_m'\|_\infty\leq M_fM_g$, and
$\|J_g^{\mathsf T}(I_N\otimes D)\|_\infty\leq\|D\|_\infty M_g$.
Thus $\|\nabla^2\mathcal U_m\|_\infty\leq K$ throughout the cube,
uniformly in $m,w$.

Taylor's formula along the segment from $Z$ to $SZ$, using absolute
values and the vanishing linear term, now gives
\[
 \Delta E\leq
 |\mathcal U_m(SZ)-\mathcal U_m(Z)|
 \leq\frac12\|\operatorname{vec}(SZ-Z)\|_1
       K\|\operatorname{vec}(SZ-Z)\|_\infty
 \leq\frac{dK}{2w}.
\]
This contradicts $w>dK/(2\Delta E)$, independently of $n$.
Hence the one-sided comparison converges to zero, as does the original
averaged profile. Finally $x^+=f(A_ng(z))\to0$ because $f(0)=g(0)=0$.
All maps, primitives, and constants in this argument are evaluated at
the fixed channel. In particular, the boundary is implemented by the
absence of exterior $f$-sites, without using $f(x;0)=0$ or comparing
different channel parameters.
\end{proof}

\section{Constituent Potentials and Fixed Points}\label{sec:potential}

The original three-message Z-side recursion has successful state $(0,0,\epsilon)$ and therefore is not a fixed-channel potential system with successful origin.  We first retain its primitive identities for computing fixed-point values, and then give the reduced system used for coupling.  For \eqref{eq:ic-z-de}, put
\[
      D_Z=\operatorname{diag}(j_Z,k,1),
\]
and define
\begin{align}
  F_Z(\bm{y}_Z;\epsilon)
     &=\hat{a}^{j_Z}\hat{b}^k+\epsilon\hat c,                 \label{eq:ic-z-F}\\
  G_Z(\bm{x}_Z)
     &=j_Z\left(a-\frac{1-(1-a)^k}{k}\right)
       +kb-(1-c)\left(1-(1-b)^k\right).                       \label{eq:ic-z-G}
\end{align}
Then \(\nabla F_Z(\bm{y}_Z;\epsilon)=D_Z\bm{f}_Z(\bm{y}_Z;\epsilon)\) and \(\nabla G_Z(\bm{x}_Z)=D_Z\bm{g}_Z(\bm{x}_Z)\), where \(\bm{f}_Z\) is the right side of \eqref{eq:ic-z-de} and \(\bm{g}_Z\) is defined by \eqref{eq:joint-hat-a}--\eqref{eq:joint-hat-c}.  The Z-side potential is
\begin{equation}
  U_Z(\bm{x}_Z;\epsilon)
    =\bm{g}_Z(\bm{x}_Z)^T D_Z\bm{x}_Z
       -G_Z(\bm{x}_Z)-F_Z(\bm{g}_Z(\bm{x}_Z);\epsilon).       \label{eq:ic-z-potential}
\end{equation}
For the X-side system \eqref{eq:ic-x-de}, put
\[
      D_X=\operatorname{diag}(j_X,k),
\]
and define
\begin{align}
  F_X(\bm{y}_X;\epsilon)
     &=\frac{j_X}{k}\hat d^k+\epsilon\hat e^k,                 \label{eq:ic-x-F}\\
  G_X(\bm{x}_X)
     &=j_Xd+ke+(1-d)^{j_X}(1-e)^k-1.                           \label{eq:ic-x-G}
\end{align}
Then \(\nabla F_X(\bm{y}_X;\epsilon)=D_X\bm{f}_X(\bm{y}_X;\epsilon)\) and \(\nabla G_X(\bm{x}_X)=D_X\bm{g}_X(\bm{x}_X)\), where \(\bm{f}_X\) is the right side of \eqref{eq:ic-x-de} and \(\bm{g}_X\) is defined by \eqref{eq:joint-hat-d}--\eqref{eq:joint-hat-e}.  The X-side potential is
\begin{equation}
  U_X(\bm{x}_X;\epsilon)
    =\bm{g}_X(\bm{x}_X)^T D_X\bm{x}_X
       -G_X(\bm{x}_X)-F_X(\bm{g}_X(\bm{x}_X);\epsilon).        \label{eq:ic-x-potential}
\end{equation}
Every constituent fixed point is a stationary point of its potential.  Boundary stationary points can also occur: for $j_X\geq3$, $g_X'(1,e)=0$, whereas $(1,e)$ is a DE fixed point only when $e=\epsilon$.

\begin{proposition}[Reduced Z-side system]\label{prop:reduced-z-system}
Fix $0\leq\epsilon<1$, and put $q_\epsilon(t)=\epsilon+(1-\epsilon)t$ and $h(t)=1-(1-t)^{k-1}$.  On the state $(a,b)$ define
\begin{align}
 \widetilde g_Z(a,b)&=(h(a),h(b))^{\mathsf T},\notag\\
 \widetilde f_Z(u,v;\epsilon)&=(u^{j_Z-1}q_\epsilon(v)^k,
                                  u^{j_Z}q_\epsilon(v)^{k-1})^{\mathsf T}.
 \label{eq:reduced-z-maps}
\end{align}
For $j_Z\geq2$, these maps form a fixed-channel potential system with
\begin{align}
 \widetilde D_Z&=\operatorname{diag}(j_Z,k(1-\epsilon)),\notag\\
 \widetilde F_Z(u,v;\epsilon)&=u^{j_Z}q_\epsilon(v)^k,\notag\\
 \widetilde G_Z(a,b;\epsilon)&=j_Z\left(a-\frac{1-(1-a)^k}{k}\right)
 +(1-\epsilon)\left(kb-1+(1-b)^k\right).
 \label{eq:reduced-z-primitives}
\end{align}
Its potential satisfies, throughout $[0,1]^2$,
\begin{equation}
 \widetilde U_Z(a,b;\epsilon)=U_Z((a,b,\epsilon)^{\mathsf T};\epsilon).
 \label{eq:reduced-z-potential-equality}
\end{equation}
\end{proposition}
\begin{proof}
Both maps preserve the cube and are monotone, and their values at zero vanish.  Differentiating \eqref{eq:reduced-z-primitives} gives the two required primitive identities.  Moreover $q_\epsilon(h(b))=1-(1-\epsilon)(1-b)^{k-1}$, so the reduced recursion is exactly the $(a,b)$ recursion with $c=\epsilon$.  Expanding its potential gives
\[
 j_Za\hat a-j_Za+\frac{j_Z}{k}[1-(1-a)^k]
 +kb\hat b-kb+(1-\epsilon)[1-(1-b)^k]-\hat a^{j_Z}\hat b^k,
\]
which is \eqref{eq:ic-z-potential} at $c=\epsilon$.  The diagonal matrix is positive for every fixed $\epsilon<1$; its channel dependence is permitted by \Cref{thm:standard-vector-saturation}.
\end{proof}

We now define the fixed-point classes used in the two constituent potential
thresholds.  The terminology follows the bounded-degree MN/HA potential
analysis in \cite{ObataJianKasaiPfister2013}; it classifies solutions of the
fixed-point equations \eqref{eq:ic-z-de} and \eqref{eq:ic-x-de}, and does not
add a channel assumption.

\begin{definition}[Successful, trivial, and nontrivial constituent fixed points]\label{def:constituent-fixed-point-classes}
Fix \(\epsilon\).  Let \(\mathcal F_Z(\epsilon)\) be the fixed-point set of
\eqref{eq:ic-z-de}.  The Z-side successful and trivial fixed-point sets are
\[
      \mathcal S_Z(\epsilon)
      =\{\bm{x}_Z\in\mathcal F_Z(\epsilon):a=b=0\},
      \qquad
      \mathcal T_Z(\epsilon)
      =\{\bm{x}_Z\in\mathcal F_Z(\epsilon):a=b=1\}.
\]
Here \(c=\epsilon\) at every Z-side fixed point.  The Z-side nontrivial
fixed-point set is
\[
      \mathcal N_Z(\epsilon)
      =\mathcal F_Z(\epsilon)\setminus
        \bigl(\mathcal S_Z(\epsilon)\cup\mathcal T_Z(\epsilon)\bigr).
\]
The Z-side non-successful fixed-point set is
\[
      \mathcal F_Z^\star(\epsilon)
      =
      \mathcal F_Z(\epsilon)\setminus\mathcal S_Z(\epsilon)
      \qquad
      \left(=\mathcal T_Z(\epsilon)\cup\mathcal N_Z(\epsilon)\right).
\]

Similarly, let \(\mathcal F_X(\epsilon)\) be the fixed-point set of
\eqref{eq:ic-x-de}.  The X-side successful and trivial fixed-point sets are
\[
      \mathcal S_X(\epsilon)
      =\{\bm{x}_X\in\mathcal F_X(\epsilon):d=e=0\},
      \qquad
      \mathcal T_X(\epsilon)
      =\{\bm{x}_X\in\mathcal F_X(\epsilon):d=1,\ e=\epsilon\}.
\]
The X-side nontrivial fixed-point set is
\[
      \mathcal N_X(\epsilon)
      =\mathcal F_X(\epsilon)\setminus
        \bigl(\mathcal S_X(\epsilon)\cup\mathcal T_X(\epsilon)\bigr).
\]
The X-side non-successful fixed-point set is
\[
      \mathcal F_X^\star(\epsilon)
      =
      \mathcal F_X(\epsilon)\setminus\mathcal S_X(\epsilon)
      \qquad
      \left(=\mathcal T_X(\epsilon)\cup\mathcal N_X(\epsilon)\right).
\]
\end{definition}

The successful fixed-point sets represent decoding success and are excluded
from the energy-gap minimization.  The trivial fixed-point sets are saturated
erasure branches whose potentials vanish at the corresponding constituent
hashing-bound values.  The nontrivial fixed-point sets contain all remaining
fixed points.

With the convention \(\inf\emptyset=+\infty\), define
\begin{equation}
   \Delta E_Z(\epsilon)
      =\inf_{\bm{x}_Z\in\mathcal F_Z^\star(\epsilon)}
          U_Z(\bm{x}_Z;\epsilon),
      \qquad
   \epsilon_{\mathrm{pot},Z}
      =\sup\{\epsilon_0\in[0,1]:\Delta E_Z(\epsilon)>0\ \text{for all }0\leq\epsilon<\epsilon_0\}.  \label{eq:ic-z-potential-threshold}
\end{equation}
Similarly, define
\begin{equation}
   \Delta E_X(\epsilon)
      =\inf_{\bm{x}_X\in\mathcal F_X^\star(\epsilon)}
          U_X(\bm{x}_X;\epsilon),
      \qquad
   \epsilon_{\mathrm{pot},X}
      =\sup\{\epsilon_0\in[0,1]:\Delta E_X(\epsilon)>0\ \text{for all }0\leq\epsilon<\epsilon_0\}.  \label{eq:ic-x-potential-threshold}
\end{equation}
The threshold used for the five-message recursion is
\begin{equation}
   \epsilon_{\mathrm{pot}}
      :=
      \min\{\epsilon_{\mathrm{pot},Z},
             \epsilon_{\mathrm{pot},X}\}.
      \label{eq:joint-potential-threshold}
\end{equation}

\begin{lemma}[Isolation of success below the candidate thresholds]\label{lem:isolated-success}
For $2\leq j_Z<k$ and $0\leq\epsilon<j_Z/k$, the successful fixed point of the reduced Z-side system is isolated.  For $k\geq3$, the X-side successful fixed point is isolated for every $\epsilon$.  Consequently, at a fixed such channel, positive potential at every nonzero fixed point implies a strictly positive energy gap.
\end{lemma}
\begin{proof}
At the reduced Z-side origin, the recursion Jacobian is zero if $j_Z\geq3$.  If $j_Z=2$, its only possible nonzero eigenvalue is $(k-1)\epsilon^k<1$: indeed $(k-1)(2/k)^k<1$ for $k\geq3$.  For the X side the Jacobian at the origin is zero when $k\geq3$.  Thus $I-T'(0)$ is nonsingular for both recursion maps $T$, and the origin is an isolated solution of $x-T(x)=0$.  The remaining fixed-point set is closed and bounded away from the origin, hence compact.  A continuous strictly positive potential on that set has positive minimum.
\end{proof}

\section{Exact Constituent Potential Positivity}\label{sec:constituent-positivity}

We prove positivity for the finite degree range $2\leq j<k\leq60$.  The proof treats every physical nontrivial fixed point, including all erasure parameters in $[0,1]$.  This stronger channel range permits the complementary-channel transformation on the HA side.

\subsection{Classical Degree Correspondence and Existing Results}\label{sec:classical-correspondence}

In the classical MN notation $(\ell,r,g)$ of \cite{OkazakiKasai2014}, the BEC recursion is
\begin{align}
 x_1^+&=\left[1-(1-x_1)^{r-1}(1-x_2)^g\right]^{\ell-1},\notag\\
 x_2^+&=\epsilon\left[1-(1-x_1)^r(1-x_2)^{g-1}\right]^{g-1}.
 \label{eq:classical-mn-correspondence}
\end{align}
Setting $(x_1,x_2)=(d,e)$ and $(\ell,r,g)=(k,j_X,k)$ gives \eqref{eq:ic-x-de}.  Its primitives and diagonal weight are precisely $F_X,G_X,D_X$ in \Cref{sec:potential}; hence $U_X$ is the classical MN potential with the same normalization.

For the Z side, eliminating $c=\epsilon$ gives the classical HA recursion with outer LDPC degrees $(j_Z,k)$ and square LDGM degrees $(k,k)$.  In HA notation $(\ell,r,g)$ this is $(j_Z,k,k)$, whose dual MN degree tuple is $(k,j_Z,k)$ \cite{KasaiSakaniwa2011,ObataJianKasaiPfister2013}.  The reduction in \Cref{prop:reduced-z-system} preserves its potential.  Thus both constituent positivity questions are classical.

More explicitly, at a Z-side fixed point put $u=\hat a$, $v=\hat b$, and $\delta=1-\epsilon$.  Then $(1-u,1-v)$ is a fixed point of the dual MN recursion at channel $\delta$: its check messages are $(1-a,1-b)$, whose variable updates return $(1-u,1-v)$.  Let $U_{\rm MN}^{(k,j_Z,k)}$ denote its classical potential with the normalization used above.  Expanding both potentials at these corresponding fixed points gives
\begin{equation}
 U_Z((a,b,\epsilon)^{\mathsf T};\epsilon)
 =U_{\rm MN}^{(k,j_Z,k)}(1-u,1-v;1-\epsilon)
   +\frac{j_Z}{k}-\epsilon.
 \label{eq:classical-ha-mn-duality}
\end{equation}
For $\epsilon<j_Z/k$, the dual channel is above $1-j_Z/k$.  A duality-based proof must therefore control the corresponding MN branch at this complementary channel; positivity only below the dual MN threshold would not suffice.  Positivity on all physical nontrivial MN fixed-point branches does suffice, because the last term in \eqref{eq:classical-ha-mn-duality} is then positive.

Positivity and capacity-achieving spatial coupling have already been proved for particular bounded-degree classical families.  Obata et al.\ prove the result for $(\ell,2,2)$ MN and $(2,\ell,2)$ HA ensembles, $\ell\geq3$ \cite{ObataJianKasaiPfister2013}.  Okazaki and Kasai prove nontrivial fixed-point potential positivity for $(\ell,3,3)$ MN ensembles, $\ell\geq3$ \cite[Sec. III--IV]{OkazakiKasai2014}.  Fukushima et al.\ explicitly collect the BEC positivity results for $(r,g)=(2,2)$ and $(3,3)$ and extend the corresponding potential-threshold result to generalized erasure channels \cite[Lemma 2 and Theorem 2]{FukushimaOkazakiKasai2015}.

For $(j_Z,j_X,k)=(4,8,12)$, the X constituent is MN $(12,8,12)$ and the Z constituent is HA $(4,12,12)$, with dual MN tuple $(12,4,12)$.  The cited $(\ell,2,2)$ and $(\ell,3,3)$ families alone do not cover these tuples.  The exact verification below covers them and all other degree pairs in the stated finite range.

\subsection{A Completed Exact Positivity Certificate}\label{sec:exact-certificate}

\begin{theorem}[Classical MN positivity at explicit finite degrees]\label{thm:certified-mn-positivity}
For every pair of integers $2\leq j<k\leq60$, the classical MN $(k,j,k)$ potential satisfies
\[
 U_{\rm MN}^{(k,j,k)}(d,e;\eta)>10^{-6}
\]
at every physical fixed point with $0<d<1$, for every $\eta\in[0,1]$.  Here physical means $d,e,\eta\in[0,1]$.  This statement is established by the exact rational certificate supplied with the source and the verification described below.
\end{theorem}

\begin{proof}
Put $u=\hat d$, $t=1-e$, and $h=\hat e$.  The fixed-point equations give
\begin{align}
 d&=u^{k-1},\qquad
 t^k=\frac{1-u}{(1-u^{k-1})^{j-1}},\qquad
 e=\eta h^{k-1},\notag\\
 A(u)&=(1-u)(1-u^{k-1}),\qquad h=1-A(u)/t.
 \label{eq:certified-mn-parameter}
\end{align}
For $0<d<1$, we have $0<u<1$ and $t>0$; $e=1$ would force $d=1$.  Substitution of $\eta h^k=eh$ into the potential gives the exact identity
\begin{align}
 U_{\rm MN}^{(k,j,k)}&=t+B(u)-\frac{(k-1)A(u)}{t},\notag\\
 B(u)&=(k-2)-(k-2)u-(k+j-2)u^{k-1}
       +\left(k+j-2-\frac jk\right)u^k.
 \label{eq:certified-mn-potential}
\end{align}
This parameterization includes all physical fixed points with $0<d<1$, without a stability assumption.

First exclude an interval next to zero.  Since $t\leq1$, the check message satisfies $h\leq1-A(u)\leq u+d$.  Also $e\leq h^{k-1}$, because $\eta\leq1$.  The first check equation and the elementary union bound therefore give
\[
 u\leq(j-1)d+ke
 \leq(j-1)u^{k-1}+k(u+u^{k-1})^{k-1}.
\]
Every physical fixed point with $u>0$ consequently satisfies
\begin{equation}
 1\leq u^{k-2}\left[j-1+k(1+u^{k-2})^{k-1}\right].
 \label{eq:certificate-left-exclusion}
\end{equation}
The right side is increasing in $u$.  The certificate records a positive dyadic rational $a_{j,k}\leq1/2$ at which it is strictly less than one, and the verifier checks this inequality exactly.  Hence $u>a_{j,k}$.

For $j\geq3$, write $S_k(u)=1+u+\cdots+u^{k-2}$.  The condition $t\leq1$ in \eqref{eq:certified-mn-parameter} requires
\[
 1\leq(1-u)^{j-2}S_k(u)^{j-1}.
\]
Set $b_{j,k}=1-[2(k-1)^2]^{-1}$.  If $u\geq b_{j,k}$, the right side is at most
\[
 (1-b_{j,k})^{j-2}(k-1)^{j-1}
   =\frac{(k-1)^{3-j}}{2^{j-2}}\leq\frac12,
\]
a contradiction.  Thus $u<b_{j,k}$.  For $j=2$, use $b_{2,k}=1$ and the cancellation
\[
 t^k=S_k(u)^{-1}.
\]
This expression extends continuously to $u=1$ with $S_k(1)=k-1$, so interval bounds can include that endpoint without division by zero.

For each rational subinterval $[L,R]\subset[a_{j,k},b_{j,k}]$, define
\begin{align}
 T_{L,R}&=\frac{1}{(1-L)^{j-2}S_k(R)^{j-1}},\notag\\
 A_{L,R}&=(1-L)(1-L^{k-1}),\notag\\
 B_{L,R}&=(k-2)-(k-2)R-(k+j-2)R^{k-1}
       +\left(k+j-2-\frac jk\right)L^k.
 \label{eq:certificate-cell-bounds}
\end{align}
Monotonicity of the factors gives $t^k\geq T_{L,R}$, $A(u)\leq A_{L,R}$, and $B(u)\geq B_{L,R}$.  The certificate stores a positive integer $m$ for which
\begin{equation}
 \tau=\frac{m}{2^{40}}>0,\qquad
 \tau^k\leq T_{L,R},\qquad
 \tau+B_{L,R}-\frac{(k-1)A_{L,R}}{\tau}>10^{-6}.
 \label{eq:certificate-exact-test}
\end{equation}
Each comparison uses only rational arithmetic, equivalently integer inequalities after clearing positive denominators.  Since $t+B-(k-1)A/t$ is increasing in $t>0$ and decreasing in $A\geq0$, \eqref{eq:certificate-exact-test} implies the claimed potential bound throughout the interval.

The supplied certificate contains every pair $2\leq j<k\leq60$, exactly once, and a finite partition of $[a_{j,k},b_{j,k}]$ for each pair.  Intervals are encoded by a dyadic subdivision depth and index.  The independent verifier reconstructs their rational endpoints, checks coverage without gaps or overlaps, verifies \eqref{eq:certificate-left-exclusion}, and checks \eqref{eq:certificate-exact-test} on every cell.  All 1711 degree pairs and all 384895 rational cells pass these exact checks.  The largest dyadic subdivision depth is 11.
  These completed checks, together with the endpoint exclusions, establish the bound at every physical nontrivial fixed point.
\end{proof}

\paragraph{Reproducibility of the exact proof.}
The exact certificate, independent verifier, and reproduction data are available in the \href{https://github.com/kasaikenta/mn-ha-css-erasure-certificates}{public GitHub repository} (\href{https://github.com/kasaikenta/mn-ha-css-erasure-certificates/tree/4d030d1faee2986a6e6a4cd4e14f1a87229ce31f}{commit \texttt{4d030d1faee2}}) and included in the source package.  The verifier \path{exact/verify_finite_family.py} uses only integer and rational arithmetic, without importing the generator or invoking a numerical root finder.  The interval inequalities certify the full continuous domains.  The repository README supplies SHA-256 checksums and commands to regenerate and independently verify the certificate.

The theorem is restricted to the checked degree range.  Extending the upper degree limit requires further exact verification or an analytic bound uniform in the degrees.

\section{Seeded Convergence and X/Z Equal-Rate Families}\label{sec:threshold-saturation}

The exact MN bound and the HA--MN duality identity give the constituent potential thresholds for every degree pair in the verified range.  The MN fixed points with $d=0$ or $d=1$ are $(d,e)=(0,0)$ and $(1,\eta)$: $d=0$ forces $e=0$, while $d=1$ forces $e=\eta$ and $e=1$ forces $d=1$.  Points with $e=0$ and $0<d<1$ are included in \Cref{thm:certified-mn-positivity}.  On the HA side, the fixed-point equations force $b=0$ when $a=0$ and $b=1$ when $a=1$; every other fixed point below $\epsilon=1$ has $0<a,b,\hat a,\hat b<1$ and therefore maps to a nontrivial dual MN fixed point.

\begin{theorem}[Z-side potential threshold]\label{thm:z-constituent-ic-potential}
For integers \(2\leq j_Z<k\leq60\),
\[
      \epsilon_{\mathrm{pot},Z}=\frac{j_Z}{k}.
\]
\end{theorem}

\begin{proof}
The successful fixed point has \(a=b=0\) and \(c=\epsilon\); it has \(U_Z=0\) and is excluded from \(\mathcal F_Z^\star(\epsilon)\).  The Z-side trivial fixed point has \(a=b=1\) and \(c=\epsilon\).  Then \(\hat a=\hat b=\hat c=1\), and direct substitution into \eqref{eq:ic-z-potential} gives
\begin{align*}
  U_Z((1,1,\epsilon)^T;\epsilon)
    &= j_Z+k+\epsilon
       -\left(j_Z-\frac{j_Z}{k}+k-1+\epsilon\right)
       -(1+\epsilon) \\
    &= \frac{j_Z}{k}-\epsilon .
\end{align*}
Thus the trivial Z-side fixed point has positive potential for \(\epsilon<j_Z/k\) and reaches zero at \(\epsilon=j_Z/k\).

Every nontrivial HA fixed point maps under \eqref{eq:classical-ha-mn-duality} to a physical nontrivial MN $(k,j_Z,k)$ fixed point at channel $1-\epsilon$.  \Cref{thm:certified-mn-positivity} therefore gives
\[
 U_Z(\bm x_Z;\epsilon)>10^{-6}+\frac{j_Z}{k}-\epsilon
 \qquad(\bm x_Z\in\mathcal N_Z(\epsilon)).
\]
Together with the trivial branch, this gives $\Delta E_Z(\epsilon)=j_Z/k-\epsilon>0$ below $j_Z/k$.  At and above $j_Z/k$, the trivial fixed point has nonpositive potential, which directly excludes a larger potential threshold.  Hence $\epsilon_{\mathrm{pot},Z}=j_Z/k$; no MAP converse is needed for this potential-threshold calculation.
\end{proof}

The X-side threshold is obtained by the same logic, with the complementary
degree ratio \(1-j_X/k\) replacing the Z-side ratio \(j_Z/k\).

\begin{theorem}[X-side potential threshold]\label{thm:x-constituent-ic-potential}
For integers \(2\leq j_X<k\leq60\),
\[
      \epsilon_{\mathrm{pot},X}=1-\frac{j_X}{k}.
\]
\end{theorem}

\begin{proof}
The successful fixed point \(d=e=0\) has \(U_X=0\) and is excluded.  The X-side trivial fixed point has \(d=1\) and \(e=\epsilon\).  Then \(\hat d=\hat e=1\), and
\begin{align*}
  U_X((1,\epsilon)^T;\epsilon)
    &= j_X+k\epsilon-(j_X+k\epsilon-1)
       -\left(\frac{j_X}{k}+\epsilon\right) \\
    &= 1-\frac{j_X}{k}-\epsilon .
\end{align*}
Thus the trivial X-side fixed point reaches zero at $\epsilon=1-j_X/k$.  By \Cref{thm:certified-mn-positivity}, every $\bm x_X\in\mathcal N_X(\epsilon)$ has $U_X(\bm x_X;\epsilon)>10^{-6}$.  Hence
\[
 \Delta E_X(\epsilon)\geq\min\{10^{-6},1-j_X/k-\epsilon\}>0
 \qquad(\epsilon<1-j_X/k).
\]
The trivial fixed point has nonpositive potential for every $\epsilon\geq1-j_X/k$, excluding a larger potential threshold.  Therefore $\epsilon_{\mathrm{pot},X}=1-j_X/k$.

\end{proof}

Consequently,
\begin{equation}
   \epsilon_{\mathrm{pot}}
      =
      \min\left\{\frac{j_Z}{k},\,1-\frac{j_X}{k}\right\}.       \label{eq:ic-product-potential-value}
\end{equation}

The five-message recursion separates into the Z-side and X-side constituent recursions.  Thus no additional five-message potential is needed.  The convergence proof below uses the reduced Z-side system in \Cref{prop:reduced-z-system} and the original two-dimensional X-side system.  Under the X/Z equal-rate condition \(j_Z+j_X=k\), the threshold in \eqref{eq:ic-product-potential-value} equals the hashing-bound parameter computed in \Cref{prop:rate}.

\begin{remark}[Role of the exact certificate]\label{rem:kasai-existing}
Classical positivity for the present degrees is established by \Cref{thm:certified-mn-positivity} and \eqref{eq:classical-ha-mn-duality}.  The convergence proof uses the resulting positive energy gaps.  No remainder nonnegativity assumption is needed.
\end{remark}

The preceding two threshold computations identify the two constituent
thresholds.  The next theorem packages them into the single threshold relevant
to the original five-message CSS recursion and records the simplification under
the X/Z equal-rate condition.

\begin{theorem}[Constituent potential threshold]\label{prop:constituent-potential-threshold}
For integers \(2\leq j_Z<j_X<k\leq60\),
\[
      \epsilon_{\mathrm{pot}}
      =
      \min\left\{\frac{j_Z}{k},\,1-\frac{j_X}{k}\right\}.
\]
If the additional X/Z equal-rate degree condition \(j_Z+j_X=k\) is imposed, then
\[
      \epsilon_{\mathrm{pot}}
       =\epsilon_{\mathrm{hash}}
       =\frac{1-R_Q^{\mathrm{des}}}{2}.
\]
In particular, $(j_Z,j_X,k)=(4,8,12)$ has $\epsilon_{\mathrm{pot}}=\epsilon_{\mathrm{hash}}=1/3$.
\end{theorem}

\begin{proof}
By the definition \eqref{eq:joint-potential-threshold} and
\Cref{thm:z-constituent-ic-potential,thm:x-constituent-ic-potential},
\[
      \epsilon_{\mathrm{pot}}
      =
      \min\left\{\frac{j_Z}{k},\,1-\frac{j_X}{k}\right\}.
\]
This proves the general threshold formula.  Under \(j_Z+j_X=k\), the two entries in the minimum are both \(j_Z/k\).  Since \(R_Q^{\mathrm{des}}=(j_X-j_Z)/k=(k-2j_Z)/k\),
\[
      \frac{j_Z}{k}
      =
      \frac{1-R_Q^{\mathrm{des}}}{2}
      =
      \epsilon_{\mathrm{hash}} .
\]
This proves the claimed equality.
\end{proof}

The exact certificate in \Cref{sec:exact-certificate} covers both constituents for every triple in \Cref{prop:constituent-potential-threshold}.

For the X/Z equal-rate example \((j_Z,j_X,k)=(4,8,12)\),
\Cref{fig:ic-potential-fixed-points} plots the constituent potentials on the
trivial fixed-point branches and on the numerically found nontrivial
fixed points.  The trivial branches cross zero at the proved potential threshold \(1/3\).  The nontrivial branches illustrate the strict positivity established by the exact certificate and duality.

\begin{figure}[!htbp]
\centering
\includegraphics[width=0.95\textwidth]{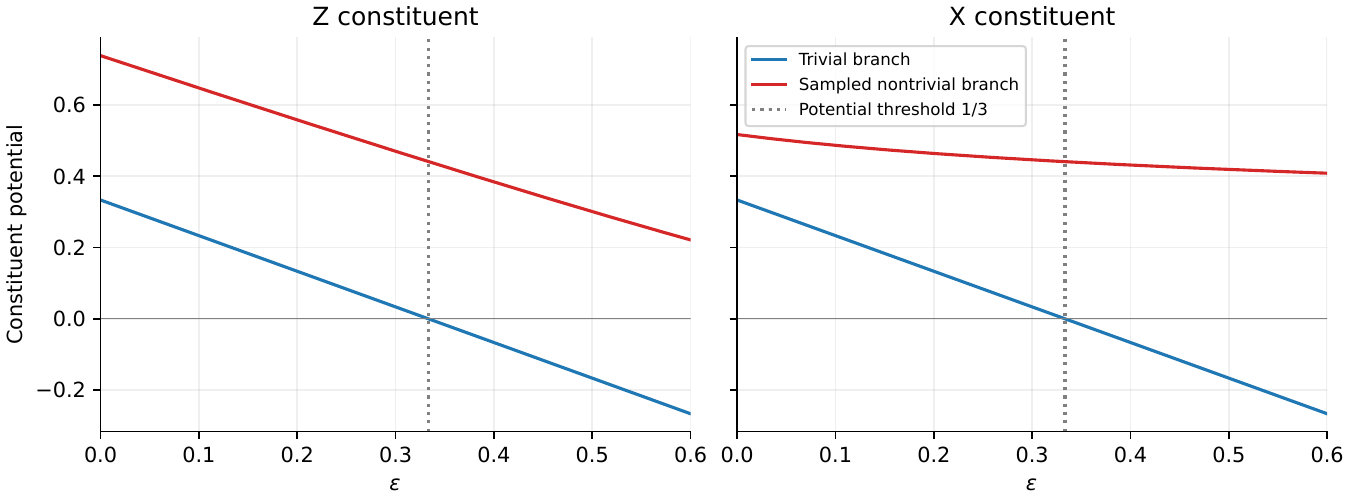}
\caption{Constituent potentials evaluated at fixed points for
\((j_Z,j_X,k)=(4,8,12)\).  The blue curve is the potential on the trivial
fixed-point branch.  The red curve is the potential on the nontrivial
fixed-point branch in
\(\mathcal N_Z(\epsilon)\) or \(\mathcal N_X(\epsilon)\), namely the fixed
points other than the successful and trivial ones in
\Cref{def:constituent-fixed-point-classes}.  The dotted vertical line marks the
proved potential threshold $1/3$.  The left and right panels are the Z and X constituents, respectively.}
\label{fig:ic-potential-fixed-points}
\end{figure}

\FloatBarrier
\Needspace{10\baselineskip}
\begin{theorem}[Convergence with an ideal message seed]\label{thm:joint-saturation}
Let $2\leq j_Z<j_X<k\leq60$ be integers.  Fix
\[
 0\leq\epsilon<\min\{j_Z/k,1-j_X/k\}.
\]
Consider the sectionwise recursion \eqref{eq:native-z-checks}--\eqref{eq:native-x-checks}, with $\epsilon_i=\epsilon$ off an interval $\mathcal S$ of at least $w$ consecutive sections and $\epsilon_i=0$ on it.  Clamp $a_i,b_i,d_i,e_i$ to zero on $\mathcal S$, and initialize these messages to one elsewhere.  For sufficiently large finite $w=w(\epsilon)$, both visible residuals tend to zero for every finite $L>|\mathcal S|$.

The same conclusion holds for the abstract common-window recursion that uses the averaging formulas \eqref{eq:general-sc-x-average}--\eqref{eq:general-seeded-sc-de} with the original three-message Z-side and two-message X-side maps, with all five messages clamped on $\mathcal S$.  On its unseeded sections the successful Z-side state is $(0,0,\epsilon)$, not the zero vector.
\end{theorem}
\begin{proof}
Theorems \ref{thm:z-constituent-ic-potential} and \ref{thm:x-constituent-ic-potential} give positive constituent energy gaps at the fixed $\epsilon$.  Use \Cref{thm:standard-vector-saturation} with $(\widetilde f_Z,\widetilde g_Z,\widetilde D_Z)$ and $(f_X,g_X,D_X)$; \eqref{eq:reduced-z-potential-equality} identifies the first energy gap with $\Delta E_Z$.  Both comparison systems therefore converge to their two-dimensional origins for sufficiently large $w$.

For the Z-side comparison, retain the same clamp on $(a,b)$ but set the channel factor to $\epsilon$ at every B-check, including checks in $\mathcal S$.  Because $\epsilon_i\leq\epsilon$, its B-check erasure update is no smaller than \eqref{eq:native-z-checks}.  Moreover, averaging commutes with the affine map $q_\epsilon$, so this comparison recursion is exactly the common-window coupling of \eqref{eq:reduced-z-maps}.  Monotonicity bounds the native $(a,b)$ profile by this comparison profile at every iteration.  Hence $a_i,b_i\to0$ and $r_{Z,i}\to0$.

Reflecting section indices converts the native X-side averaging orientation to \Cref{def:seeded-sc-de}.  The reflected seed still has width at least $w$.  Its off-seed recursion is $(f_X,g_X)$ and thus converges to zero, giving $r_{X,i}\to0$.

For the abstract Z-side recursion, after its first update $c_i=\epsilon_i$, and its incoming averaged channel component is at most $\epsilon$.  The same comparison therefore applies from that time, using the largest $(a,b)$ initialization.  Its residual is $\epsilon_i(\mathcal M_+\hat c)_i$ and tends to zero.  The abstract X-side recursion is already in the common-window orientation.  Choose $w$ large enough for both comparison systems.
\end{proof}

This is an achievability statement for every channel strictly below the proved potential threshold.  It does not assert an exact threshold for a fixed finite $(L,w)$, nor does it account for the physical cost of the ideal message seed.

The X/Z equal-rate specialization gives explicit design rates and hashing-bound parameters throughout the verified degree range.

\begin{proposition}[Explicit X/Z equal-rate families]\label{prop:rate-param-xz-equal-rate-saturation}
For integers $j\geq2$ and $\lambda\geq1$ with $2j+\lambda\leq60$, the degree triple
\[
 (j_Z,j_X,k)=(j,j+\lambda,2j+\lambda)
\]
has
\[
 R_Q^{\mathrm{des}}=\frac{\lambda}{2j+\lambda},\qquad
 \epsilon_{\mathrm{pot}}=\epsilon_{\mathrm{hash}}=\frac{j}{2j+\lambda}.
\]
For every $0\leq\epsilon<j/(2j+\lambda)$, its ideally seeded coupled DE converges to zero residual for sufficiently large coupling width.  The example $(4,8,12)$ corresponds to $j=4$, $\lambda=4$, and $R_Q^{\mathrm{des}}=1/3$.
\end{proposition}
\begin{proof}
The degree triple satisfies $j_Z+j_X=k$ and $2\leq j_Z<j_X<k\leq60$.  \Cref{prop:rate,prop:constituent-potential-threshold} give the rate and potential-threshold identities, and \Cref{thm:joint-saturation} gives convergence below that threshold.
\end{proof}

Every displayed rate representative in \Cref{fig:xz-equal-rate-family-threshold-rate} belongs to this proved finite set.  The proposition supplies explicit rates within the checked range; it does not assert that every rational design rate is represented with $k\leq60$.

\paragraph{Numerical branch checks.}
A scan from a small grid of initial guesses need not locate all nontrivial fixed points.  We instead use the fixed-point equations to parameterize each nontrivial branch.  For the Z side, let $0<u<1$ and define
\begin{align}
 a(u)&=1-(1-u)^{1/(k-1)},\qquad
 v(u)=\left[\frac{a(u)}{u^{j_Z-1}}\right]^{1/k},\notag\\
 b(u)&=u^{j_Z}v(u)^{k-1},\qquad
 E_Z(u)=1-\frac{1-v(u)}{[1-b(u)]^{k-1}}.\label{eq:numeric-z-branch}
\end{align}
Keep only points with $0<v(u)<1$ and $0\leq E_Z(u)<j_Z/k$.  For the X side,
\begin{align}
 d(u)&=u^{k-1},\qquad
 e(u)=1-\left[\frac{1-u}{[1-d(u)]^{j_X-1}}\right]^{1/k},\notag\\
 h_X(u)&=1-[1-d(u)]^{j_X}[1-e(u)]^{k-1},\qquad
 E_X(u)=\frac{e(u)}{h_X(u)^{k-1}}.\label{eq:numeric-x-branch}
\end{align}
Keep $0\leq e(u)<1$ and $0\leq E_X(u)<1-j_X/k$.  These formulas cover every nontrivial fixed point in the stated domains; a finite numerical sampling of $u$ is still not an exact sign certificate.

For example, at $(j_Z,j_X,k)=(2,3,5)$ and $\epsilon=0.2$, the X-side point
\[
 (d,e)\simeq(0.391495092044130,\ 0.108092604846116)
\]
satisfies the fixed-point equations to below $10^{-14}$ in maximum norm.  This example illustrates why failure of a grid search to return a point is not evidence of its absence.  The supplied script checks fixed-point residuals and the total potentials.  These floating-point samples illustrate the branches; the proof of positivity is the independent exact certificate in \Cref{sec:exact-certificate}.

\Cref{fig:xz-equal-rate-family-threshold-rate} displays approximately uniform design-rate representatives.  For $r_i=i/21$, $i=1,\ldots,20$, choose a triple $(j,j+\lambda,2j+\lambda)$ with $j\geq2$, $\lambda\geq1$, and $|\lambda/(2j+\lambda)-r_i|\leq0.02$, minimizing $k=2j+\lambda$ (and then $j$ to break ties).  Both plotted coordinates are computed from the design formulas.  Their location on $R=1-2\epsilon$ is an algebraic identity, rather than a measurement of coupled decoding thresholds.

\begin{figure}[!htbp]
\centering
\includegraphics[width=0.78\textwidth]{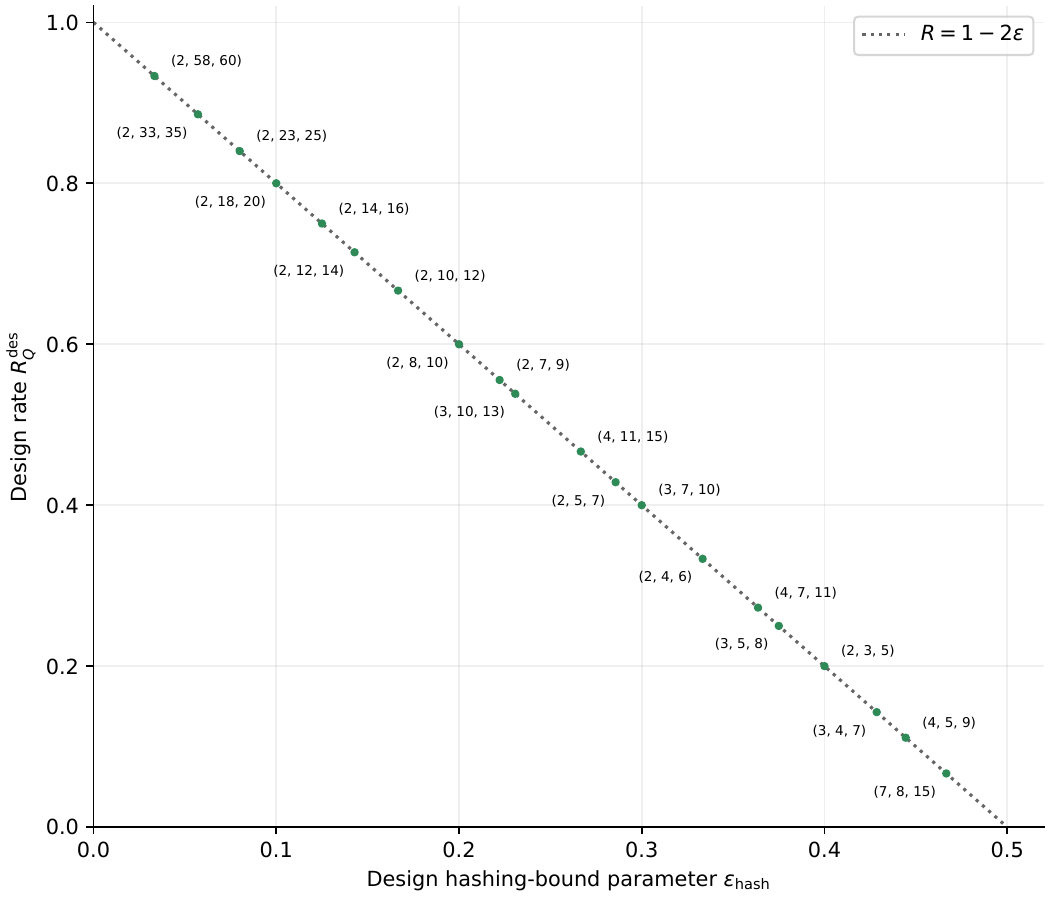}
\caption{Low-degree X/Z equal-rate degree triples used to display approximate
rate coverage in the design-parameter plane.  The representatives are chosen
near \(r_i=i/21\), \(i=1,\ldots,20\), with
\(|R_Q^{\mathrm{des}}-r_i|\leq0.02\) and minimal \(k\), giving \(k\leq60\).
The horizontal axis is
\(\epsilon_{\mathrm{hash}}=j/(2j+\lambda)\), and the vertical
axis is \(R_Q^{\mathrm{des}}=\lambda/(2j+\lambda)\).  The dotted line is
\(R=1-2\epsilon\).  The points are design parameters, not measured BP thresholds.}
\label{fig:xz-equal-rate-family-threshold-rate}
\end{figure}

\Needspace{8\baselineskip}
\paragraph{Numerical DE example.}
For $(j_Z,j_X,k)=(4,8,12)$, the proved potential threshold and the design hashing-bound parameter are both $1/3$.  We compare the native sectionwise recursion with the abstract common-window recursion of \Cref{thm:joint-saturation}, at $L=1024$, $w=16$, $\mathcal S=\{0,\ldots,15\}$, and $\epsilon=0.3325$.  The messages $a,b,d,e$ are zero on the seed and initialized to one elsewhere.  In the native recursion $c_i=\epsilon_i$ from initialization.  For the abstract recursion $c_i$ starts at one off the seed and is reset to $\epsilon_i$ after the first update.  Seed messages remain clamped.

The implementation uses double precision, simultaneous updates, and cumulative sums for window averages.  Residuals are evaluated from the current check messages before the variable update; convergence is declared when the maximum of both residual profiles is below $10^{-12}$.  The abstract recursion met this criterion at iteration 240,239; the native sectionwise recursion met it at iteration 240,239.  Both runs used the same tolerance and initialization.

The abstract Z-side residual is $\epsilon_i(\mathcal M_+\hat c)_i$; the native residual is the local value in \eqref{eq:native-z-variables}.  Figures \ref{fig:seeded-ic-z-wave} and \ref{fig:seeded-ic-x-wave} show the abstract recursion to retain the original numerical comparison.  The curve marked iteration zero is the initial channel-erasure profile, and the remaining curves show every $10000$ iterations together with the final recorded iteration.  The script and machine-readable outputs record both coupling conventions.

\begin{figure}[!htbp]
\centering
\includegraphics[width=0.98\textwidth]{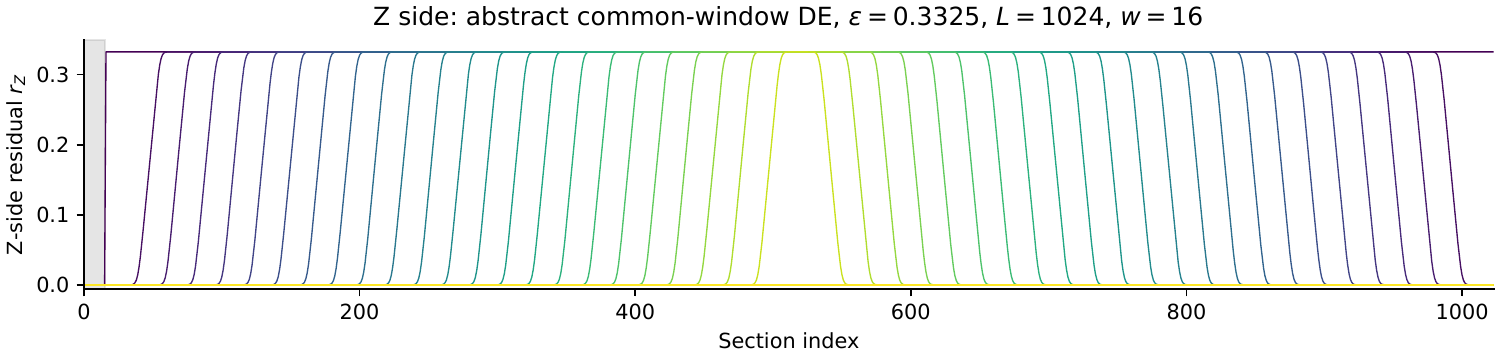}
\caption{Z-side abstract common-window seeded DE for \((j_Z,j_X,k)=(4,8,12)\),
\(L=1024\), \(w=16\), \(|\mathcal S|=w\), and \(\epsilon=0.3325\).  The
vertical axis is the abstract residual $r_{Z,i}=\epsilon_i(\mathcal M_+\hat c)_i$.  The plotted curves start
with the initial profile, followed by every \(10000\) iterations and the final recorded iteration.}
\label{fig:seeded-ic-z-wave}
\end{figure}

\begin{figure}[!htbp]
\centering
\includegraphics[width=0.98\textwidth]{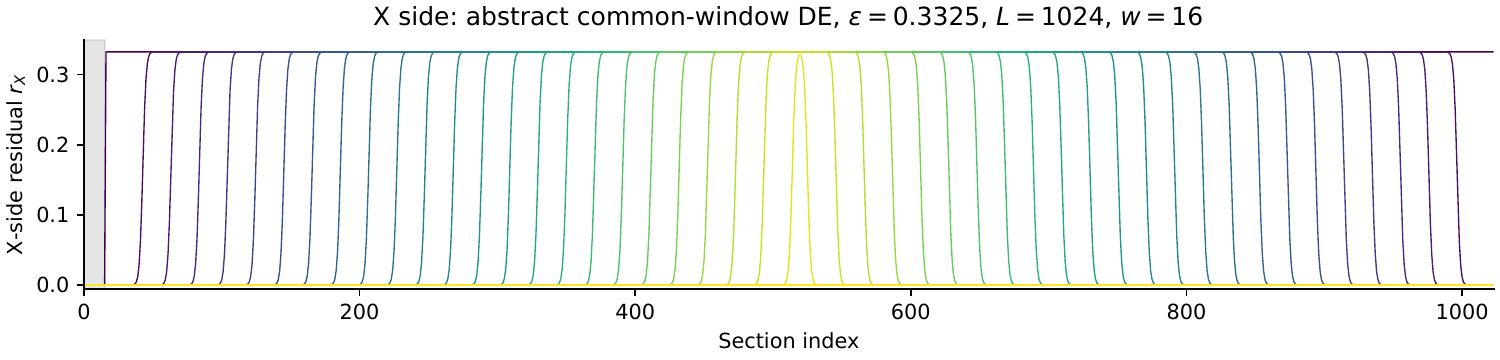}
\caption{X-side abstract common-window seeded DE for the same parameters as
\Cref{fig:seeded-ic-z-wave}.  The vertical axis is the coupled residual \(r_{X,i}=\epsilon_i(\mathcal M_+\hat e)_i^k\).  The curves show the initial profile, every \(10000\) iterations, and the final recorded iteration.}
\label{fig:seeded-ic-x-wave}
\end{figure}

\FloatBarrier

\section{Conclusion and Future Work}\label{sec:conclusion-future}
The sparse nested representation gives a CSS pair with generally dense visible-coordinate check matrices.  Its hard-erasure DE consists of two constituent systems sharing the erasure locations.  For every $2\leq j_Z<j_X<k\leq60$, their potential thresholds are $j_Z/k$ and $1-j_X/k$, as established by the completed exact MN positivity certificate and classical HA--MN duality.  The reduced Z-side representation has successful origin and the same fixed-point potential as the original three-message description.  A fixed-channel shift argument and a monotone comparison then prove ideal-seed convergence below the smaller constituent value, for both the native sectionwise DE and the abstract common-window DE.  Under $j_Z+j_X=k$, the proved potential threshold equals $(1-R_Q^{\mathrm{des}})/2$.

The proof covers all 1711 constituent degree pairs with $2\leq j<k\leq60$, including the main example $(4,8,12)$ and all displayed rate representatives.  The source package includes the full rational certificate and an independent verifier.  Extending the result to unbounded degrees requires additional exact certificates or a uniform analytic positivity argument.  The numerical example at $\epsilon=0.3325$ illustrates the proved convergence for $(4,8,12)$.

The seed clamps the auxiliary messages as well as the visible messages.  Setting only $\epsilon_i=0$ supplies a weaker boundary condition and does not force the punctured messages to zero.  A finite-code implementation must provide the required known auxiliary information while preserving CSS commutation and the syndrome interface.  Its effect on the actual rank, rate, and seed cost must be established.  For comparison with the unseeded design rate, the asymptotic regime should in particular have $w/L\to0$.

The DE residuals are marginal per-coordinate probabilities.  A finite-length logical block-success theorem further requires concentration or a suitable replacement argument, a joint scaling of section size, chain length, width and iteration count, and control of residual ambiguity modulo the CSS stabilizer.  Classical bit-error and block-error distinctions are discussed in \cite{LentmaierTruhachevZigangirovCostello2005}; their quantum analogue remains a finite-length question here.

\bibliographystyle{IEEEtran}
\bibliography{references}
\end{document}